\documentclass[12pt]{article}
\usepackage{booktabs}
\usepackage{latexsym, amsbsy, amssymb,multirow, epsfig, amsmath}
\usepackage{hyperref} 
\usepackage[section]{placeins}
\usepackage{mathrsfs}
\usepackage{amsmath}
\usepackage{soul}

\usepackage{algorithm,algorithmic}
\def\beqr{\begin{eqnarray}}
\def\eeqr{\end{eqnarray}}
\def\beqrs{\begin{eqnarray*}}
\def\eeqrs{\end{eqnarray*}}

\hypersetup{
	breaklinks=true,
}
\usepackage[page,title]{appendix}

\usepackage{color}
\makeatletter
\renewcommand\normalsize{%
\abovedisplayskip 7\p@ \@plus2\p@ \@minus7\p@
\belowdisplayskip \abovedisplayskip
\let\@listi\@listI}
\makeatother

\usepackage{natbib,graphicx,setspace,lscape,longtable}
\usepackage{natbib,epsfig,graphicx,epstopdf}
\usepackage{amsmath,amsthm,amssymb,color}
\RequirePackage[mathlines, displaymath]{lineno}
\bibpunct{(}{)}{;}{a}{,}{,}

\def\beqr{\begin{eqnarray}}
\def\eeqr{\end{eqnarray}}

\numberwithin{equation}{section}
\renewcommand{\baselinestretch}{1.75}

\newtheorem{theo}{Theorem}[section]

\newtheorem{proposition}{Proposition}[section]

\newtheorem{lemm}{Lemma}[section]

\newtheorem{remark}{Remark}

\newtheorem{condition}{Condition}

\usepackage{authblk}

\usepackage{latexsym}
\usepackage{epsfig}
\usepackage{bm}
\usepackage{algorithm,algorithmic}
\usepackage{threeparttable}
\usepackage{graphicx}
\usepackage[small]{caption2}
\usepackage{threeparttable}
\usepackage{dcolumn}
\usepackage{multirow}
\usepackage{booktabs,epstopdf}
\usepackage{xcolor}

\usepackage{xurl}

\begin{document}
	
\def\spacingset#1{\renewcommand{\baselinestretch}%
{#1}\small\normalsize} \spacingset{1}
		
\title{\bf  
Random Projection Tests via Cauchy Combination for  Two-Sample Mean
}

\author[1]{Yanyan Ouyang}
\author[2]{Ruoxi Peng}
\author[3]{Wangli Xu}
\author[2]{Tao Qiu*}

\affil[1]{School of Sciences, Chang'an University, Xi'an 710064, China}
\affil[2]{Center for Statistics and Data Science, Beijing Normal University, Zhuhai 519087, China}
\affil[3]{Center for Applied Statistics and School of Statistics, Renmin University of China, Beijing 100872, China}

\maketitle

\begin{abstract}
High-dimensional two-sample mean testing is challenging when the dimension exceeds the sample size. 
The random projection method proposed by \cite{Lopes2011} addresses this difficulty by mapping the data to a lower dimension space where Hotelling's \(T^2\) statistic can be applied, while retaining useful covariance information and gaining power when the variables exhibit non-negligible covariance structure.
However, single projection tests may be sensitive to the realized projection matrix, whereas existing multiple projection procedures often rely on resampling or simulation for calibration, with limited theoretical understanding. 
Moreover, projection-based Hotelling tests may be less sensitive to sparse mean differences. 
To address these limitations, we propose a Cauchy-combined random projection test (CRPT), which 
applies Hotelling's \(T^2\) test after multiple independent random projections and combines the projected \(p\)-values through the Cauchy transformation.
The proposed method retains the ability of random projection methods to incorporate covariance information while reducing reliance on any single projection. 
Under the Gaussian assumption, we establish the null tail behavior of the proposed statistic and further investigate its asymptotic power under suitable alternatives.
To improve sensitivity to sparse alternatives, we further develop a power-enhanced version of CRPT. 
Simulation studies and real data analysis are conducted to examine the performance and practical applicability of the proposed procedures.
\end{abstract}


\spacingset{1.75} 

\section{Introduction}

Testing differences in mean vectors between two groups  is a fundamental problem in statistics.
In modern biomedical studies, especially genomic research, the number of features often far exceeds the sample size. In such settings, the classical  Hotelling’s $T^2$ test statistics may become invalid, motivating the development of methods tailored for high‑dimensional two‑sample mean testing. 
Specifically, let $\{\mathbf{x}_1, \ldots, \mathbf{x}_{n_1}\}$ and $\{\mathbf{y}_1, \ldots, \mathbf{y}_{n_2}\}$ be two independent samples from $p$-dimensional distributions with the mean vectors $\boldsymbol{\mu}_1 = (\mu_{11},\ldots,\mu_{1p})^\top$ and $\boldsymbol{\mu}_2 = (\mu_{21},\ldots,\mu_{2p})^\top$, respectively. The hypotheses of interest are given by
\begin{align}\label{eq:globalh0}
H_0: \boldsymbol{\mu}_1 = \boldsymbol{\mu}_2 \quad \text{versus} \quad H_1: \boldsymbol{\mu}_1 \neq \boldsymbol{\mu}_2.
\end{align}
When $p$ is fixed, the Hotelling’s $T^2$ test \citep{Hotelling1931} is widely applied with the test statistic
\begin{align}\label{eq:hotelling}
T^2 = \frac{n_1 n_2}{n_1 + n_2} (\bar{\mathbf{x}} - \bar{\mathbf{y}})^\top \mathbf{S}^{-1} (\bar{\mathbf{x}} - \bar{\mathbf{y}}),
\end{align}
where $\bar{\mathbf{x}} = n_1^{-1}\sum_{i=1}^{n_1} \mathbf{x}_i$, $\bar{\mathbf{y}} = n_2^{-1}\sum_{j=1}^{n_2} \mathbf{y}_j$, and 
$\mathbf{S} = \{\sum_{i=1}^{n_1} (\mathbf{x}_i - \bar{\mathbf{x}})(\mathbf{x}_i - \bar{\mathbf{x}})^\top + \sum_{j=1}^{n_2} (\mathbf{y}_j - \bar{\mathbf{y}})(\mathbf{y}_j - \bar{\mathbf{y}})^\top\}/(n_1+n_2-2)$ is the pooled sample covariance matrix.
When $p>n_1 + n_2 - 2$, 
the matrix $\mathbf{S}$ is singular, so the inverse $\mathbf S^{-1}$ appearing in \eqref{eq:hotelling} is not well defined and the classical Hotelling's \(T^2\) statistic cannot be directly adopted.

This difficulty has motivated a line of research that refines Hotelling's \(T^2\) statistic for high-dimensional data.
\citet{Bai1996} proposed a test statistic by replacing the sample covariance matrix $\mathbf{S}$ with the identity matrix $\mathbf{I}_p$. Subsequently, \citet{Chen2010} developed a method based on the $U$ statistic and relaxed the restrictions on the relationship between the dimension and the sample sizes.
Related methods include diagonal covariance estimates $\operatorname{diag}(\mathbf{S})$ \citep{Srivastava2008}, ridge regularization $\mathbf{S}+\lambda\mathbf{I}_p$ \citep{chen2011regularized}, refined \(L_2\)-type methods \citep{Zhang2020}, or nonparametric techniques \citep{ouyang2022}. Many of the above methods replace the sample covariance matrix with simplified structures and can be effective when the covariance matrix is weakly correlated or nearly diagonal. 
However, they may lose power when the covariance matrix has non-negligible dependence or some dominant eigenvalues, whereas projection-based tests provide an alternative way to retain covariance information in high-dimensional two-sample mean testing; related developments include \citet{qiu2021} and \citet{chen2025}.

The random projection method proposed by \citet{Lopes2011} projects the \(p\)-dimensional observations onto a lower-dimensional space, where Hotelling's \(T^2\) statistic can be applied. 
Specifically, for a random projection matrix $\mathbf{P} \in \mathbb{R}^{k \times p}$ with i.i.d. standard normal entries and $k<p$, the hypotheses are
\begin{align}
H_{0,\mathbf{P}}: \mathbf{P} \boldsymbol{\mu}_1 = \mathbf{P} \boldsymbol{\mu}_2 \quad \text{versus} \quad H_{1, \mathbf{P}}: \mathbf{P} \boldsymbol{\mu}_1 \neq \mathbf{P} \boldsymbol{\mu}_2.
\end{align}
The corresponding projected Hotelling statistic is
\begin{align}
	T_{\mathbf{P}}^2 = \frac{n_1n_2}{n_1+n_2} (\bar{\mathbf{x}} - \bar{\mathbf{y}})^\top \mathbf{P}^\top (\mathbf{P} \mathbf{S} \mathbf{P}^\top)^{-1} \mathbf{P} (\bar{\mathbf{x}} - \bar{\mathbf{y}}).
\end{align}
Although this single-projection procedure is computationally convenient and theoretically tractable, its performance may be sensitive to the particular random projection matrix $\mathbf{P}$. 
To reduce the randomness induced by a single projection, several multiple-projection methods have been proposed, such as those in \citet{thulin2014}, \citet{zhangpro2016}, and \citet{srivastava2016}. 
These methods aggregate the projected statistics obtained from multiple independent projection matrices.
However, The projected test statistics $T_{\mathbf{P}_1}^2$ and $T_{\mathbf{P}_2}^2$ computed from different projection matrices $\mathbf{P}_1$ and $\mathbf{P}_2$ are dependent, 
because they share the same sample mean difference \(\bar{\mathbf x}-\bar{\mathbf y}\) and the same pooled covariance matrix \(\mathbf S\). 
Consequently, the null distribution of the aggregated statistic is difficult to derive analytically. 
Existing multiple-projection procedures therefore often rely on permutation or other simulation-based methods to determine critical values, which may become computationally demanding when the number of projections or dimension is large.

This issue motivates the use of a combination rule whose null behavior remains tractable under certain assumptions.
The Cauchy combination test proposed by \citet{cct} provides such a tool. 
In the random projection setting, each projection produces a projected \(p\)-value \(p_b\), \(b=1,\ldots,B\). 
We combine these projected \(p\)-values by applying the Cauchy transformation and averaging the transformed values as
\begin{align}
T(p_1,\ldots,p_B)  = \frac{1}{B}\sum_{i=1}^{B} \tan\{\pi(1/2-p_i)\}.
\end{align}
A key feature of this transformation is that very small projected \(p\)-values are mapped to large positive values, and hence dominate the upper tail of the combined statistic.
Under suitable tail dependence conditions, the combined statistic has the same upper-tail behavior as a standard Cauchy random variable, even when the projected \(p\)-values are dependent. 
Therefore, to apply the Cauchy combination to multiple projection test, we need to verify such tail dependence conditions for the projected statistics. 
A necessary, though not sufficient, step in our argument is to establish an asymptotic tail-independence property for the projected statistics.
Specifically, under null hypothesis, for two independent projection matrices \(\mathbf P_1\) and \(\mathbf P_2\), we need to verify that $\mathbb P(T_{\mathbf{P}_1}^2>t,T_{\mathbf{P}_2}^2>t)	=o(\mathbb P(T_{\mathbf{P}_1}^2>t))$.
The related tail-independence principle has also been studied and generalized in subsequent work; see, for example, \citet{2023Fang}.


The above Hotelling-type and quadratic-form test statistics perform well for dense signals, but may lose power under sparse alternatives. 
To detect both dense and sparse signals, $L_\gamma$-type methods have been developed \citep{xuAdaptive2016, he2021}. Meanwhile, the power enhancement framework, proposed by \citet{2015Fan} for high-dimensional hypothesis testing, provides a useful tool for addressing this issue and has been widely adopted in recent studies; see, for example, \citet{yu2023,zhang2024,wangCross2024}. 
Its main idea is to add an auxiliary statistic that is asymptotically negligible under null hypothesis but diverges under certain sparse alternatives. 
Motivated by this idea, we further incorporate a power enhancement component into the multiple-projection statistic to improve its sensitivity to sparse mean differences. 

The contribution of this work is threefold. 
\begin{itemize}
\item [] 1.  We propose a Cauchy combined random projection test for high-dimensional two-sample mean problems. 
The proposed procedure applies Hotelling's \(T^2\) test after multiple independent random projections and combines the projected \(p\)-values through the Cauchy transformation. 
Thus, it retains the ability of projection tests to incorporate covariance information while reducing sensitivity to any single projection.

\item [] 2. We provide a tail analysis for the dependence induced by repeated random projections computed from the same data. 
In particular, we establish a pairwise asymptotic tail-independence property for the projected statistics, which promises the standard Cauchy upper-tail approximation for the combined statistic. 

\item [] 3. We study the asymptotic power of both the individual projected tests and the combined statistic. 
Since projection-based Hotelling tests are mainly driven by dense signals, we further introduce a power-enhanced version of the proposed test to improve its sensitivity to sparse alternatives. 
\end{itemize}

The remainder of this paper is organized as follows. Section \ref{sec:method} introduces the proposed combined random projection test and establishes its tail behavior. In Section \ref{sec:power}, we analyze the asymptotic power of both the individual projected statistics and the combined test. 
Section \ref{sec:pe} presents a power-enhanced version of the test and its theoretical properties for sparse alternatives. 
Simulation results comparing our methods with existing approaches under various dependence structures and distributions are reported in Section \ref{sec:simulation}. 
Section \ref{sec:realdata} applies the proposed tests to real genomic datasets. 
Finally, Section \ref{sec:conclusion} concludes with a discussion of the findings and outlines potential future research. All the technical details are provided in the Appendix.

\section{The Combined Random Projection Test}
\label{sec:method}

In this section, we develop the Cauchy combined random projection test for the high-dimensional two-sample mean problem. 
We first construct projected Hotelling statistics and the corresponding projected \(p\)-values, and then aggregate these \(p\)-values through the Cauchy transformation to construct the combined projection test statistic.
We then derive a standard Cauchy upper-tail approximation for the combined statistic, leading to a  analytic rejection rule for the proposed test.

\subsection{The Cauchy combined statistic}

We are interested in testing \eqref{eq:globalh0}
in the high-dimensional setting where $p \gg n_1 + n_2$.
In this subsection, we propose an ensemble approach that combines the results of multiple random projections. 
Let $n = n_1 + n_2 - 2$. 
For $b = 1, \ldots, B$, generate independent random projection matrices $\mathbf{P}_{b}\in  \mathbb{R}^{k \times p}$ with $1 \le k < n$, the projected hypotheses are
\begin{align}\label{eq:h0pb}
	H_{0,b}: \mathbf{P}_{b}\,\boldsymbol{\mu}_1 = \mathbf{P}_{b}\,\boldsymbol{\mu}_2 \quad \text{versus} \quad H_{1,b}: \mathbf{P}_{b}\, \boldsymbol{\mu}_1 \neq \mathbf{P}_{b}\, \boldsymbol{\mu}_2,
\end{align}
and the corresponding projected Hotelling statistics are
\begin{align}\label{eq:tkb2}
	T_{b}^2 = \frac{n_1 n_2}{n_1 + n_2} [\mathbf{P}_{b}(\bar{\mathbf{x}} - \bar{\mathbf{y}})]^\top (\mathbf{P}_{b} \mathbf{S} \mathbf{P}_{b}^\top)^{-1} [\mathbf{P}_{b}(\bar{\mathbf{x}} - \bar{\mathbf{y}})].
\end{align}
Suppose that $\{\mathbf{x}_1, \ldots, \mathbf{x}_{n_1}\}$ and $\{\mathbf{y}_1, \ldots, \mathbf{y}_{n_2}\}$ follow $N(\boldsymbol{\mu}_1, \boldsymbol{\Sigma})$ and $N(\boldsymbol{\mu}_2, \boldsymbol{\Sigma})$, respectively, where $\boldsymbol{\Sigma}$ is positive definite.
When $p \ge n$, the pooled covariance matrix $\mathbf{S}$ has rank $n$ almost surely.
Since $k < n$, the matrix $\mathbf{P}_{b} \mathbf{S} \mathbf{P}_{b}^\top$ is almost surely non-singular. 
Conditional on $\mathbf{P}_{b}$, the projected data $\mathbf{P}_{b} \mathbf{x}_i$ and $\mathbf{P}_{b} \mathbf{y}_j$ follow $N(\mathbf{P}_{b} \boldsymbol{\mu}_1, \mathbf{P}_{b} \boldsymbol{\Sigma} \mathbf{P}_{b}^\top)$ and $N(\mathbf{P}_{b} \boldsymbol{\mu}_2, \mathbf{P}_{b} \boldsymbol{\Sigma} \mathbf{P}_{b}^\top)$, respectively. 
Under null hypothesis in \eqref{eq:h0pb}, $(n-k+1)T_{b}^2/(nk) $ conditioned on $\mathbf{P}_{b}$ follows $F(k,n-k+1)$, where $F(k,n-k+1)$ denotes the $F$ distribution with degrees of freedom $k$ and $n-k+1$.
The corresponding $p$-value $p_b$ for \eqref{eq:h0pb} is
\begin{align}\label{eq:pb}
p_b = 1-F_{k,n-k+1}\Big(\frac{n-k+1}{nk}T_{b}^2\Big),
\end{align}
where $F_{k,n-k+1}(\cdot)$ represents the cumulative distribution function of $F(k,n-k+1)$.  
Following \citet{cct}, we aggregate these $p$-values by applying the Cauchy combination test and the combined random projection test statistic is constructed as 
\begin{align}
T_{\text{CRPT}} = \frac{1}{B} \sum_{b=1}^B \tan\left\{ \pi \left( 0.5 - p_b \right) \right\}.
\end{align}
We reject the null hypothesis when $T_{\text{CRPT}}$ is large. The tangent transform maps small projected $p$-values to large positive numbers, making the combined statistic larger and the null hypothesis more likely to be rejected.
At the same time, averaging over $B$ independent projections reduces the variability caused by the random projection.

\subsection{Behavior of the statistic {\color{black}under the null}}

In this subsection, we establish the null tail behavior of the proposed combined random projection statistic. 
The Cauchy combination method of \citet{cct} shows that the upper tail of a Cauchy combined statistic can be approximated by the tail of a standard Cauchy random variable. 
However, such an approximation generally requires {\color{black}asymptotic tail independence among the $p$-values}, related ideas have been further investigated by \citet{2023Fang}. 
Therefore, in the multiple random projections setting, we first study the pairwise asymptotic {\color{black}tail independence} property of {\color{black}$T_i^2$ and $T_j^2$, where the statistics are} obtained from two independent random projection matrices {\color{black}$\mathbf{P}_{i}$ and $\mathbf{P}_{j}$, $1\le i\not= j\le B$}. 
Based on this tail property, we then derive the standard Cauchy {\color{black}upper tail} approximation for the combined statistic under null hypothesis. 
This approximation allows us to adopt the upper $\alpha$-quantile of the standard Cauchy distribution as the critical value for $T_{\rm CRPT}$.
We begin with the following conditions.

\begin{condition}\label{cond:normal}
Assume that $\{\mathbf{x}_1, \ldots, \mathbf{x}_{n_1}\}$ and $\{\mathbf{y}_1, \ldots, \mathbf{y}_{n_2}\}$ are two independent samples following $p$-dimensional multivariate normal distributions $N(\boldsymbol{\mu}_1, \boldsymbol{\Sigma})$ and $N(\boldsymbol{\mu}_2, \boldsymbol{\Sigma})$, respectively, where $\boldsymbol{\Sigma}$ is a $p \times p$ positive definite covariance matrix.
\end{condition}

\begin{condition}\label{cond:projection}
The random projection matrices {\color{black}$\mathbf{P}_{1}, \ldots, \mathbf{P}_{B}\in \mathbb{R}^{k\times p}$} are independent and have independent standard normal entries with $k<n$.
\end{condition}

\begin{condition}[Bounded condition]\label{cond:eigenvalues}
Assume that there exist the constants $c_1$ and $c_2$ such that $0<c_1 \le\lambda_{\min}(\mathbf \Sigma) \le \lambda_{\max}(\mathbf \Sigma)\le c_2 <\infty$, where $\lambda_{\min}(\mathbf \Sigma)$ and $\lambda_{\max}(\mathbf \Sigma)$ denote the minimum and maximum eigenvalues of $\mathbf \Sigma$, respectively. 
\end{condition}

{\color{black}Conditions~\ref{cond:normal}--\ref{cond:projection} are imposed to obtain the exact null distribution of the projected statistic $T_b^2$ and the corresponding \(p\)-value $p_b$ for $1\le b \le B$.} 
Condition \ref{cond:eigenvalues} on the eigenvalues of the covariance matrix is {\color{black}an assumption widely adopted} in the high dimensional setting; {\color{black}see, for example,} \cite{2014cai}. 
Under these conditions, we obtain the following theorem.

\begin{theo}\label{theo:t_tail}
	Suppose Conditions \ref{cond:normal}-\ref{cond:eigenvalues} hold. 
	Under null hypothesis in \eqref{eq:h0pb}, for fixed $n$,  as $t\to\infty$ satisfying $\log t = o(p)$, we have
	\begin{align}\label{theo:titj}
		\mathbb P(T_{i}^2>t,T_{j}^2>t)
		=o(\mathbb P(T_{i}^2>t)),\quad 1\le i\neq j\le B.
	\end{align}
\end{theo}

Theorem~\ref{theo:t_tail} establishes the asymptotic tail independence between projected Hotelling statistics $T_{i}^2$ and $T_{j}^2$ for $1 \le i \ne j \le B$.
{\color{black}The result ensures that the probability of joint tail events is asymptotically negligible relative to the marginal tail probability of a single projected statistic.}
Based on Theorem~\ref{theo:t_tail}, we derive the following theorem regarding the asymptotic tail behavior of the combined statistic $T_{\rm CRPT}$.

\begin{theo}\label{theo:null1}
	Suppose Conditions \ref{cond:normal}-\ref{cond:eigenvalues} hold. 
	Under null hypothesis, for fixed $B$ and $n$, as $u\to\infty$ and $p/\log u\to\infty$, we have
	\begin{align}
		\frac{\mathbb{P}(T_{\rm CRPT} > u)}{\mathbb{P}(W_C > u)} \to 1, 
	\end{align}
	where $W_C$ represents a standard Cauchy random variable. 
\end{theo}

Theorem \ref{theo:null1} shows that the upper tail of $T_{\rm CRPT}$ {\color{black} can be approximated by} the tail of a 
standard Cauchy variable. 
Therefore, the corresponding $p$-value approximation for $T_{\rm CRPT}$ is
$p_{\rm CRPT} =0.5 - \arctan(T_{\rm CRPT})/\pi$,
and then we reject the null hypothesis when $p_{\rm CRPT}\le\alpha$ or $T_{\rm CRPT} \ge \tan\{\pi(0.5 - \alpha)\}$ with the significance level $\alpha$.

\section{Asymptotic power analysis}\label{sec:power}

In this section, we analyze the asymptotic power of the proposed tests under alternatives. Firstly, we derive the asymptotic power of a single projected test based on $T_{b}^2$ under the alternative hypothesis. 
We then extend the analysis to the combined statistic $T_{\rm CRPT}$, demonstrating that the combined test preserves the asymptotic detection ability of projected Hotelling tests while combining evidence across multiple random projections.

\subsection{Asymptotic statistical power of $T_{b}^2$}

Let $\boldsymbol\delta := \boldsymbol \mu_1 - \boldsymbol \mu_2$ denote the shift vector. 
For projection $\mathbf{P}_{b}$, the twice the Kullback-Leibler divergence between the projected sampling distributions $N(\mathbf{P}_{b} \boldsymbol{\mu}_1, \mathbf{P}_{b} \boldsymbol{\Sigma} \mathbf{P}_{b}^\top)$ and $N(\mathbf{P}_{b} \boldsymbol{\mu}_2, \mathbf{P}_{b} \boldsymbol{\Sigma} \mathbf{P}_{b}^\top)$ {\color{black}is}
\begin{align}\label{eq:delta_kb}
\Delta_{b}^2:= \boldsymbol\delta^\top \mathbf{P}_{b}^\top (\mathbf{P}_{b} \mathbf{\Sigma} \mathbf{P}_{b}^\top)^{-1} \mathbf{P}_{b}\boldsymbol\delta.
\end{align}

{\color{black}To derive the asymptotic power approximation for the projected test,} we consider the following condition. 

\begin{condition}\label{cond:c1}
	Assume that the sizes of two samples satisfies $n_1/n\to {\color{black}\pi_s} \in(0,1)$ and the dimension of the projection matrix satisfies $k/n\to {\color{black}\pi_d}\in(0,1)$ as $n\to\infty$.
\end{condition}

Condition \ref{cond:c1} indicates that the sizes of two samples and the projection dimension are in balance with the total sample size. {\color{black}We next derive the asymptotic conditional power for a single projected test, which is summarized in the following lemma.}


\begin{lemm}\label{prop:power1}
	Suppose Conditions \ref{cond:normal}-\ref{cond:c1} hold and
    ${\color{black}\pi_s(1-\pi_s)\Delta_{b}^2/\pi_d} \stackrel{p}{\to}\gamma\in[0,\infty)$.
	Then, 
	the {\color{black}conditional rejection probability} 
    satisfies that, as $n\to\infty$,  
	\begin{align}\label{eq:beta1}
        {\color{black}\mathbb P(p_b\le\alpha\mid \mathbf P_b)}
		- \Phi\bigg(\frac{z_\alpha + {\color{black}\pi_s(1-\pi_s)}\sqrt{1/(2\pi_d)-1/2}\sqrt{n}\Delta_{b}^2}{{\color{black}\sqrt{1+2\gamma + \pi_d \gamma^2}}}\bigg)
		\stackrel{p}{\to} 0,
	\end{align}
	Here, $\Phi(\cdot)$ denotes the cumulative distribution function of the standard normal distribution and $\Phi(z_\alpha) = \alpha$.
\end{lemm}

{\color{black}Lemma} \ref{prop:power1} establishes the asymptotic power conditional on the random projection $\mathbf P_{b}$. 
The approximation in \eqref{eq:beta1} indicates that the conditional power is governed by the projected signal strength \(\sqrt n\,\Delta_b^2\).
The conditional asymptotic power {\color{black}converges to one in probability} when $\sqrt n\,\Delta_{b}^2\xrightarrow{p}\infty$.
In contrast, if $\sqrt n\,\Delta_{b}^2 \xrightarrow{p}0$, 
the conditional rejection probability converges to the significance level $\alpha$ in probability.
To establish the behavior of $\Delta_{b}^2$, we derive the following lemma.

\begin{lemm}\label{lemm:delta_projection_bound}
	Suppose Condition \ref{cond:projection} holds and $\mathbf{\Sigma}$ is positive definite. Assume that
	$\boldsymbol\delta\in\mathbb R^p$ is a vector satisfying $\boldsymbol\delta\neq \mathbf 0$ and $k < p$. 
	Let $b_1$ and $b_2$ be constants such that $0 < b_1 < 1 < b_2$ and $b_2k/p < 1$. Then, as $k\to\infty$ and $p\to\infty$, we have
	\begin{align}
	\mathbb P\Bigg(
	\frac{b_1 k}{p \lambda_{\max}(\boldsymbol\Sigma)} \|\boldsymbol\delta\|_2^2 \le \Delta_{b}^2 \le \frac{b_2 k}{p \lambda_{\min}(\boldsymbol\Sigma)} \|\boldsymbol\delta\|_2^2
	\Bigg)\to 1.
	\end{align}
\end{lemm}

Under Condition \ref{cond:eigenvalues}, 
Lemma~\ref{lemm:delta_projection_bound} implies that {\color{black}$\Delta_{b}^2$ is of the same order as $k\|\boldsymbol\delta\|_2^2/p$}
with probability tending to one.
This result indicates that, when $k\ll p$, the projected signal strength \(\Delta_{b}^2\) can be considerably smaller than \(\|\boldsymbol\delta\|_2^2\).
Therefore, if the original signal is not sufficiently strong, the projected test may fail to reject the null hypothesis.

{\color{black}Here and below, 
we write
\(A_x\asymp B_x\) as \(x\to \infty\) if there exist constants \(0<C_1\le C_2<\infty\) and $M>0$ such that
$C_1\le A_x/B_x\le C_2$ for all $x>M$.}
For the alternative $\boldsymbol\delta^\top\mathbf \Sigma^{-1}\boldsymbol\delta = o(1)$ considered in \cite{Lopes2011}, under Condition \ref{cond:eigenvalues}, it follows that  $\|\boldsymbol\delta\|_2^2 
\asymp \boldsymbol\delta^\top\mathbf \Sigma^{-1}\boldsymbol\delta = o(1)$.
Therefore, from Lemma \ref{lemm:delta_projection_bound}, we obtain
$\sqrt{n}\Delta_{b}^2 = o_p(\sqrt{n}k/p) = o_p(1)$ for $n = O(p^{2/3})$, and then 
{\color{black}the conditional power function $\mathbb P(p_b\le\alpha\mid \mathbf P_b)$ converges to $\alpha$ in probability.}
Therefore, to derive the power property of the random projected test, we consider the following {\color{black}alternative condition, under which the signal remains detectable after projection.}
%
%

\begin{condition}[Alternative]\label{cond:alternative_relaxed}
	The shift vector \(\boldsymbol\delta=\boldsymbol\mu_1-\boldsymbol\mu_2\) satisfies $p/n^{3/2}=o(\|\boldsymbol\delta\|_2^2)$.
\end{condition}

{\color{black}Under Condition~\ref{cond:eigenvalues}, Condition~\ref{cond:alternative_relaxed} allows weaker signals than the condition considered in \citet{Chen2010}. Indeed, Condition~\ref{cond:eigenvalues} implies}
$\|\boldsymbol\delta\|_2^2 
\asymp \boldsymbol\delta^\top\mathbf \Sigma\boldsymbol\delta$ and ${\rm tr}(\mathbf \Sigma)\asymp p$. 
Therefore, the condition ${\rm tr}(\boldsymbol\Sigma)/n
=o(\boldsymbol\delta^\top\boldsymbol\Sigma\boldsymbol\delta)$ used by \citet{Chen2010} is equivalent to
$p/n=o(\|\boldsymbol\delta\|_2^2)$ under Condition~\ref{cond:eigenvalues}.
In contrast, Condition~\ref{cond:alternative_relaxed} only requires
$p/n^{3/2}=o(\|\boldsymbol\delta\|_2^2)$,
which is weaker by a factor of \(n^{1/2}\). 

We next discuss the asymptotic power of a single projected Hotelling test under Condition~\ref{cond:alternative_relaxed}. 
{\color{black}Lemma}~\ref{prop:power1} shows that the conditional power is governed by the projected signal strength \(\Delta_{b}^2\), more precisely through the term $\sqrt n\Delta_{b}^2$.
By Lemma~\ref{lemm:delta_projection_bound}, for any fixed \(0<b_1<1\),
\begin{align}
\mathbb{P}\bigg(\Delta_{b}^2
\ge
\frac{b_1 k}{p\lambda_{\max}(\boldsymbol\Sigma)}
\|\boldsymbol\delta\|_2^2\bigg)\to 1.
\end{align}
Under Conditions~\ref{cond:c1} and \ref{cond:alternative_relaxed}, we have $k/n$ {\color{black}converges to $\pi_d$} and $p/n^{3/2}=o(\|\boldsymbol\delta\|_2^2)$. Then $\sqrt n\Delta_{b}^2$ {\color{black}diverges to infinity in probability,}
which indicates that the projected signal $\Delta_{b}^2$ separates from the null scale for each single projected test. 
Together with the conditional asymptotic power in {\color{black}Lemma}~\ref{prop:power1}, this suggests that, conditional on $\mathbf P_{b}$, a single projected test should reject the null hypothesis with probability tending to one. 
Since the projection matrix \(\mathbf P_b\) is randomly generated, it is also important to characterize the rejection probability induced by both the data and the random projection.
The following theorem summarizes the conditional and unconditional rejection probabilities for a single projected test.

\begin{theo}\label{theo:power1}
	Suppose Conditions~\ref{cond:normal}-\ref{cond:c1} hold. 
	Assume that the alternative satisfies Condition \ref{cond:alternative_relaxed} and $p>n$. Then, for significance level \(\alpha\), as $n\to\infty$, the conditional power function satisfies
	\begin{align}\label{eq:theo_power11}
	\mathbb P(p_b\le\alpha\mid \mathbf P_{b})
	\stackrel p\to1.
	\end{align}
	Consequently, as $n\to\infty$, we have
	\begin{align}\label{eq:theo_power12}
	\mathbb P(p_b\le\alpha)\to1.
	\end{align}
\end{theo}

Theorem~\ref{theo:power1} establishes the power property of {\color{black}single} projected test under Condition~\ref{cond:alternative_relaxed}. 
Specifically, for any significance level $\alpha$, the conditional rejection probability given the random projection $\mathbf P_b$ converges to one in probability. 
After averaging over the randomness of the projection, the unconditional rejection probability also tends to one.


\subsection{Asymptotic statistical power of $T_{\text{CRPT}}$}

In this subsection, we analyze the power of the combined statistic $T_{\text{CRPT}}$ under the alternative hypothesis. 
{\color{black}We first provide a comparison between the combined test and a single projected test. 
Let $\mathcal D=\{\mathbf x_1,\ldots,\mathbf x_{n_1},\mathbf y_1,\ldots,\mathbf y_{n_2}\}$ denote the observed data.
Let $\phi(p):=\tan\{\pi(0.5-p)\}$ and $q_t(\mathcal D):=\mathbb P(p_b>t\mid \mathcal D)$ for $t\in(0,1)$. We consider the following condition.

\begin{condition}
\label{cond:conditional_power}
Suppose that there exist constants
\(0<\eta<\alpha<1-\zeta<1\) such that
\begin{align}\label{cond:alternativep1}
\frac{\phi(\eta)+(B-1)\phi(1-\zeta)}{B}\ge t_\alpha,
\end{align}
and
\begin{align}\label{cond:alternativep2}
q_\eta(\mathcal D)^B
+
Bq_{1-\zeta}(\mathcal D)
=
o_p\{q_\alpha(\mathcal D)\}.
\end{align}
\end{condition}
Condition~\ref{cond:conditional_power} is an additional assumption on the conditional probability \(q_t(\mathcal D)=\mathbb P(p_b>t\mid\mathcal D)\) under the alternative.
This condition is satisfied if, for example, 
there exist a sequence \(r_n\to\infty\), a constant 
\(\varepsilon>0\), and a positive continuous function \(I(t)\), strictly increasing on \((0,1)\), such that
\begin{align} 
\sup_{t\in(0,1)}
\left|
-\frac{\log q_t(\mathcal D)}{r_n}
-
I(t)
\right|
\stackrel p\to0 .
\end{align}
This formulation covers several decay rates of the conditional tail probability
\(q_t(\mathcal D)=\mathbb P(p_b>t\mid \mathcal D)\). 
For example, when \(r_n=n\), we have
$q_t(\mathcal D)=\exp\{-nI(t)+o_p(n)\}$, which corresponds to an exponential decay rate.
If \(r_n=\log n\), the same condition becomes $q_t(\mathcal D)=n^{-I(t)+o_p(1)}$, corresponding to polynomial decay.
A slower logarithmic decay rate is obtained by taking \(r_n=\log\log n\), which gives $q_t(\mathcal D)=(\log n)^{-I(t)+o_p(1)}$.

We next briefly interpret the two requirements in Condition~\ref{cond:conditional_power}. 
Equation ~\eqref{cond:alternativep1} ensures that the combined statistic reaches the rejection region whenever at least one projected \(p\)-value is no larger than \(\eta\) and no projected \(p\)-value exceeds \(1-\zeta\). Indeed, we consider the event
\[
\min_{1\le b\le B}p_b\le \eta
\quad\text{and}\quad
\max_{1\le b\le B}p_b\le 1-\zeta.
\]
Since \(\phi(p)\) is strictly decreasing on \((0,1)\), we have
\[
T_{\rm CRPT}
=
\frac1B\sum_{b=1}^B\phi(p_b)
\ge
\frac{\phi(\eta)+(B-1)\phi(1-\zeta)}{B}
\ge
t_\alpha,
\]
and then the combined test rejects the null hypothesis.
Equation~\eqref{cond:alternativep2} characterizes the conditional probability that the combined test fails to reject the null hypothesis.
The term \(q_\eta(\mathcal D)^B\) corresponds to the conditional probability that all \(B\) projections fail to produce a sufficiently small projected \(p\)-value, while \(Bq_{1-\zeta}(\mathcal D)\) accounts for the possibility that at least one projected \(p\)-value is close to one. 
The condition requires the sum of these two probabilities to be asymptotically negligible relative to \(q_\alpha(\mathcal D)\). 
This leads to the following proposition.

\begin{proposition}\label{prop:power_comparison}
Suppose Condition~\ref{cond:conditional_power} holds. For any $B=1,\ldots,B$, we have
\begin{align}
\mathbb P(T_{\rm CRPT}< t_\alpha\mid\mathcal D)=o_p\left\{
\mathbb P(p_b>\alpha\mid\mathcal D)\right\}.
\end{align}
\end{proposition}

Proposition~\ref{prop:power_comparison} indicates that, under Condition~\ref{cond:conditional_power}, the conditional probability that the combined test fails to reject the null hypothesis is of smaller order than the conditional failure probability of a single projected test. 
This gives a sharp comparison between the combined test and a single projected test. 
However, Condition~\ref{cond:conditional_power} is formulated in terms of the conditional tail probabilities \(q_t(\mathcal D)\), whose verification for a concrete mean difference \(\boldsymbol\delta\) can be difficult under the high-dimensional two-sample mean test. 
We therefore next establish the power behavior under the more directly alternative condition, Condition~\ref{cond:alternative_relaxed}.
}

\begin{theo}\label{theo:power2}
	Suppose Conditions \ref{cond:normal}-\ref{cond:c1} hold. Then, under the alternative satisfying Condition \ref{cond:alternative_relaxed},
	for significance level $\alpha\in(0,1)$ and fixed $B$, 
	\begin{align}
	\lim_{n \to \infty} \mathbb{P}(T_{\rm CRPT} \ge t_\alpha) = 1, 
	\end{align}
	where $t_\alpha$ denotes the upper $\alpha$-quantile of the standard Cauchy distribution.
\end{theo}

Theorem~\ref{theo:power2} establishes that the asymptotic statistical power of $T_{\rm CRPT}$ converges to one under Condition~\ref{cond:alternative_relaxed}. 
Recall that Condition~\ref{cond:alternative_relaxed} requires $p/n^{3/2}=o(\|\boldsymbol\delta\|_2^2)$, that is, the \(L_2\)-type signal strength \(\|\boldsymbol\delta\|_2^2\) needs to be sufficiently large. 
Therefore, \(T_{\rm CRPT}\) is mainly effective for dense alternatives, whereas it may lose power when the signal is sparse.
To address this limitation, we next propose the Cauchy random projection test {\color{black}with power enhancement}.

\section{\color{black}Combined Test with Power Enhancement}\label{sec:pe}

In this section, we develop a power-enhanced version of $T_{\text{CRPT}}$ to improve sensitivity to sparse alternatives. 
First, {\color{black}the test statistic} is constructed by augmenting \(T_{\rm CRPT}\) with an enhancement component. 
We then study its asymptotic properties. The power analysis demonstrates that the proposed test preserves the power of \(T_{\rm CRPT}\) under dense alternatives, while gaining sensitivity to sparse mean differences.

\subsection{\color{black}Test Statistic with Power Enhancement}

First, we introduce the power enhancement technique proposed by \citet{2015Fan}. 
The enhancement component is designed to be asymptotically inactive under null hypothesis, but to become active when sufficiently strong sparse signals are present.
To be specific, for \(j=1,\ldots,p\), let \(\bar x_j\) and \(\bar y_j\) denote the \(j\)-th components of the sample mean vectors \(\bar{\mathbf x}\) and \(\bar{\mathbf y}\), respectively. 
Let \(s_{jj}\) be the \(j\)-th diagonal element of the pooled sample covariance matrix \(\mathbf S\). 
Let $\delta_{n,p}=c_0\log(\log n)\sqrt{\log p}$,
where \(c_0>0\) is a constant. 
The power enhancement component is defined as
\begin{align}
J_0=\sqrt n\mathbb I\left\{
\max_{1\le j\le p}
\frac{|\bar x_j-\bar y_j|}{\sqrt{s_{jj}}}>\frac{\delta_{n,p}}
{\sqrt{n_1n_2/(n_1+n_2)}}
\right\}.
\end{align}
The power-enhanced Cauchy random projection statistic is then constructed as
\begin{align}
{\color{black}T_{\mathrm{CRPT\text{-}PE}}}=T_{\rm CRPT}+J_0.
\end{align}
The final test rejects \(H_0\) when ${\color{black}T_{\mathrm{CRPT\text{-}PE}}}\ge t_\alpha$,
where \(t_\alpha\) is the upper \(\alpha\)-quantile of the standard Cauchy distribution.

\subsection{Theoretical Properties}
In this subsection, we establish the asymptotic properties of the power enhancement component $J_0$ and the statistic \({\color{black}T_{\mathrm{CRPT\text{-}PE}}}\). 
Recall that $\boldsymbol\delta = \boldsymbol \mu_1 - \boldsymbol \mu_2$, where $\boldsymbol\delta = (\delta_1,\ldots,\delta_p)^\top$.
We consider the sparse local alternative as follows.
\begin{condition}[Sparse alternative]\label{cond:alt_sparse}
	Suppose that 
	the shift $\boldsymbol \delta$ satisfies $\max_{1\le j\le p}\{\left|\delta_j\right|/\sigma_{j j}^{1 / 2}\}>3  \delta_{n, p}\sqrt{n_1 + n_2}/ \sqrt{n_1 n_2}$, where $\delta_{n,p} = c_0 \log(\log n) \sqrt{\log p}$.
\end{condition}

Condition~\ref{cond:alt_sparse} requires that at least one coordinate has a standardized mean difference $|\delta_j|/\sigma_{j j}^{1/2}$ exceeding the detection threshold $3  \delta_{n, p}  \sqrt{n_1 + n_2}/ \sqrt{n_1 n_2}$. 
Therefore, Condition~\ref{cond:alt_sparse} is designed to capture sparse alternatives in which the signal is concentrated on a small number of coordinates. 
The following theorem establishes the behavior of the power enhancement component $J_0$ under the null and the sparse alternative.

\begin{theo}\label{theo:pe}
	Suppose Condition \ref{cond:normal} holds and $\log p = o(n)$. Under the null $H_0: \boldsymbol \mu_1 = \boldsymbol \mu_2$, as $n,p\to\infty$, we have
	\begin{align*}
	\mathbb{P}(J_0 = 0 \mid H_0) \to 1.
	\end{align*}
	Under the sparse alternative {\color{black}$H_1$} satisfying Condition \ref{cond:alt_sparse}, as $n,p\to\infty$, we have 
	\begin{align*}
	\mathbb{P}(J_0 = \sqrt n \mid {\color{black}H_1}) \to 1.
	\end{align*}
\end{theo}

Theorem~\ref{theo:pe} implies that the enhancement term $J_0$ is asymptotically negligible under null hypothesis. 
Meanwhile, \(J_0\) can capture sparse signals that may not be sufficiently reflected in the projected statistics $T_{\mathrm{CRPT}}$. Specifically, under the sparse alternative {\color{black}satisfying Condition \ref{cond:alt_sparse}, Theorem~\ref{theo:pe} implies that the probability of $\{J_0=\sqrt n\}$ converges to one.}
Consequently, \({\color{black}T_{\mathrm{CRPT\text{-}PE}}}\) diverges to {\color{black}infinity}, which enhances the ability of the test to detect sparse mean differences. 
The next theorem summarizes the power properties of \({\color{black}T_{\mathrm{CRPT\text{-}PE}}}\).

\begin{theo}\label{theo:tpe}
	Suppose Conditions \ref{cond:normal}-\ref{cond:c1} hold and $\log p = o(n)$. 
	Consider that the alternative ${\color{black}H_1}$ satisfies Condition~\ref{cond:alternative_relaxed} or \ref{cond:alt_sparse}. For any significance level $\alpha > 0$, as $n,p\to\infty$, we have 
	\begin{align*}
		\mathbb{P}({\color{black}T_{\mathrm{CRPT\text{-}PE}}} \ge t_{\alpha}\mid {\color{black}H_1}) \to 1.
	\end{align*}
\end{theo}

Theorem~\ref{theo:tpe} shows that the statistic 
\({\color{black}T_{\mathrm{CRPT\text{-}PE}}}\) can work effectively under both dense and sparse alternatives. 
{\color{black}Under dense alternatives satisfying Condition~\ref{cond:alternative_relaxed}}, the projected Cauchy random projection statistic \(T_{\rm CRPT}\) already accumulates sufficient evidence against the null hypothesis, and ${\color{black}T_{\mathrm{CRPT\text{-}PE}}}$ preserves this power. 
{\color{black}Under dense alternatives satisfying Condition~\ref{cond:alt_sparse}, \(\mathbb{P}(J_0=\sqrt n)\) converges to one}, and hence the enhancement component drives \({\color{black}T_{\mathrm{CRPT\text{-}PE}}}\) to diverge.

\begin{remark}
{\color{black}Under the Gaussian setting specified in Condition \ref{cond:normal}}, the conditional distribution of \((n-k+1)T_b^2/(nk)\) given \(\mathbf P_b\) is exactly \(F(k,n-k+1)\). 
{\color{black}Consequently, under the null hypothesis, the corresponding \(p\)-value $p_b$ in \eqref{eq:pb} is exactly uniformly distributed on \((0,1)\), conditional on \(\mathbf P_b\).}
However, this {\color{black}exact distribution} generally does not remain valid when the normality assumption is relaxed.

{\color{black}Nevertheless, for fixed $k$, $p_b$ can also be approximated by a uniform distribution on \((0,1)\) as \(n\to\infty\).
Specifically, the statistic $T_{b}^2$ has an asymptotic \(\chi_k^2\) distribution under suitable conditions} as \(n\to\infty\) for fixed \(k\); see, for example, \citet{fujikoshi1997,kakizawa2008}.
Moreover, let $F_0$ {\color{black}denote a random variable with the $F(k,n-k+1)$ distribution. Then,} for fixed $k$, $nkF_0/(n-k+1)$ {\color{black}converges weakly to} the \(\chi_k^2\) distribution as $n\to\infty$. 
These results indicate a possible direction for extending the proposed {\color{black}combined test} to non-Gaussian settings. 
For theoretical clarity, however, the null tail approximation and power analysis in this paper are developed under the Gaussian {\color{black}assumption}, while non-Gaussian settings are examined through simulations.
\end{remark}

\section{Simulation Studies}
\label{sec:simulation}

In this section, we conduct simulations to evaluate the performance of the proposed tests. We compare its empirical performance with the test methods proposed by  
\cite{Bai1996} (abbreviated as BS), 
\cite{Srivastava2013} (abbreviated as SKK), \cite{Chen2010} (abbreviated as CQ),
\cite{Lopes2011} (abbreviated as LWJ)
and \cite{2014cai} (abbreviated as CLX).
The dimension of the single random projection matrix is set as $k = \lfloor n/2 \rfloor$.
The proposed methods are denoted as CRPT and CRPT-PE with $B = 50$ random projections.

The data are generated from the following three settings.
\begin{itemize}

	\item Case 1: ({\it Long-term dependence structure})
	Data are generated from two $p$-dimensional multivariate normal distributions $N_p(\boldsymbol{\mu}_1, \boldsymbol{\Sigma}_1)$ and  $N_p(\boldsymbol{\mu}_2, \boldsymbol{\Sigma}_1)$. Let $\boldsymbol{\Sigma}_1 = (1-\rho)\mathbf{I}_p + \rho \mathbf{J}_p$, where $\mathbf{I}_p$ denotes the $p$-dimensional identity matrix and $\mathbf{J}_p$ represents a $p\times p$ matrix with all elements equal to 1.
	
	\item Case 2: ({\it Short-term dependence structure}) Data are generated from two $p$-dimensional multivariate normal distributions $N_p(\boldsymbol{\mu}_1, \boldsymbol{\Sigma}_2)$ and  $N_p(\boldsymbol{\mu}_2, \boldsymbol{\Sigma}_2)$.
	Consider the auto-regression structure with $\boldsymbol{\Sigma}_2 = (\rho^{|i-j|})_{p \times p}$ and $\rho = 0.3$. 
	
	\item Case 3:({\it non-Gaussian}) Data are generated from two $p$-dimensional multivariate $t_3$ distributions. Consider the long-term dependence structure $\boldsymbol{\Sigma}_1$ defined in Case 1.

\end{itemize}

\subsection{Performance under null hypothesis}

In this subsection, we evaluate the empirical Type I error rate of all competing methods under null hypothesis $H_0: \boldsymbol{\mu}_1 = \boldsymbol{\mu}_2 = \mathbf{0}$. 
The sample size pairs and dimensions are 
$(n_1,n_2)= \{(10,15), (15,15),(15,30), (30,30), (30,50), (50,50), (50,100), (100,100)\}$
and $p= \{100, 200, 500, 1000, 2000\}$. 
Type I error rates are estimated based on 2000 replications at significance levels $\alpha = 0.05$.

Table \ref{table:sizemodel1} presents the sizes of all the tests for Case 1 under long‑term dependence structure. The simulated results demonstrate that BS, CQ, SKK, LWJ and CRPT can maintain the type I error rates across all scenarios. In contrast, CLX exhibits notable inflation, particularly under small sample sizes or high dimension. 
Meanwhile, CRPT‑PE fails to control the Type I error rate for small sample sizes since 
$\mathbb{P}(J_0 = 0)$ is small in such settings. As the sample size increases, this probability approaches one, and CRPT‑PE regains proper control, which aligns with the theoretical results presented in Section \ref{sec:pe}.

Table \ref{table:sizemodel2} reports the empirical sizes of the seven tests for Case 2 under short‑term dependence structure. In Table \ref{table:sizemodel2}, BS, CQ, LWJ and CRPT can also maintain the type I error rates across all scenarios for Case 2.
Unlike Case~1, SKK can not maintains Type~I error. CLX remains notably inflated under small sample sizes or large $p$.
For CRPT‑PE, the same pattern persists, with poor control in small samples due to low 
$\mathbb{P}(J_0 = 0)$ and recovery as the sample size increases.

Table~\ref{table:sizemodel3} presents the sizes for Case~3 under the multivariate $t_3$ distribution with long-range dependence. Although CRPT was originally developed under normality, it can control the Type~I error well across all scenarios under the multivariate $t_3$ distribution, demonstrating its robustness to heavy-tailedness. Meanwhile, BS, CQ, LWJ and SKK can control the type I error rates, while CLX exhibits mild inflation. As before, CRPT‑PE shows poor control for small $n$ but recovers as $n$ grows, consistent with Section \ref{sec:pe}.

\begin{table}[htp]
	\centering
	\caption{\label{table:sizemodel1} The sizes of the seven tests for Case 1 under long‑term dependence structure at significance level $\alpha = 0.05$.  The last column reports the simulated $\mathbb{P}(J_0 = 0)$. }
	\centering
	\resizebox{\ifdim\width>\linewidth\linewidth\else\width\fi}{!}{
		\fontsize{10}{12}\selectfont
		\begin{tabular}[t]{lrlllllllc}
			\toprule
			($n_1$, $n_2$) & $p$ & BS & CQ & SKK & CLX & LWJ & CRPT & CRPT-PE & $\mathbb{P}(J_0 = 0)$\\
			\midrule
			& 100 & 0.0725 & 0.0715 & 0.0640 & 0.1360 & 0.0565 & 0.0625 & 0.1885 & 0.5315\\
			& 200 & 0.0660 & 0.0660 & 0.0540 & 0.1410 & 0.0530 & 0.0570 & 0.1590 & 0.5080\\
			& 500 & 0.0880 & 0.0860 & 0.0530 & 0.2065 & 0.0495 & 0.0650 & 0.1845 & 0.4415\\
			& 1000 & 0.0710 & 0.0700 & 0.0345 & 0.2185 & 0.0490 & 0.0570 & 0.1725 & 0.4185\\
			\multirow{-5}{*}{\raggedright\arraybackslash (10,15)} & 2000 & 0.0870 & 0.0865 & 0.0355 & 0.2730 & 0.0470 & 0.0500 & 0.1835 & 0.3650\\
			\cmidrule{1-10}
			& 100 & 0.0735 & 0.0735 & 0.0560 & 0.0995 & 0.0540 & 0.0570 & 0.1820 & 0.6835\\
			& 200 & 0.0775 & 0.0775 & 0.0530 & 0.1250 & 0.0525 & 0.0635 & 0.1855 & 0.6620\\
			& 500 & 0.0640 & 0.0640 & 0.0355 & 0.1655 & 0.0520 & 0.0570 & 0.1875 & 0.6360\\
			& 1000 & 0.0740 & 0.0740 & 0.0285 & 0.1760 & 0.0440 & 0.0630 & 0.1885 & 0.6180\\
			\multirow{-5}{*}{\raggedright\arraybackslash (15,15)} & 2000 & 0.0855 & 0.0855 & 0.0265 & 0.2150 & 0.0525 & 0.0500 & 0.1970 & 0.5785\\
			\cmidrule{1-10}
			& 100 & 0.0685 & 0.0665 & 0.0490 & 0.0770 & 0.0540 & 0.0675 & 0.1535 & 0.8720\\
			& 200 & 0.0675 & 0.0680 & 0.0505 & 0.0845 & 0.0450 & 0.0595 & 0.1350 & 0.8835\\
			& 500 & 0.0705 & 0.0695 & 0.0320 & 0.0975 & 0.0515 & 0.0595 & 0.1290 & 0.8920\\
			& 1000 & 0.0775 & 0.0760 & 0.0335 & 0.1240 & 0.0520 & 0.0485 & 0.1360 & 0.8760\\
			\multirow{-5}{*}{\raggedright\arraybackslash (15,30)} & 2000 & 0.0875 & 0.0885 & 0.0265 & 0.1295 & 0.0500 & 0.0600 & 0.1365 & 0.8880\\
			\cmidrule{1-10}
			& 100 & 0.0720 & 0.0720 & 0.0485 & 0.0680 & 0.0490 & 0.0755 & 0.1205 & 0.9345\\
			& 200 & 0.0745 & 0.0745 & 0.0450 & 0.0635 & 0.0525 & 0.0740 & 0.1090 & 0.9525\\
			& 500 & 0.0740 & 0.0740 & 0.0290 & 0.0870 & 0.0475 & 0.0655 & 0.1000 & 0.9550\\
			& 1000 & 0.0800 & 0.0800 & 0.0245 & 0.0810 & 0.0570 & 0.0550 & 0.0845 & 0.9615\\
			\multirow{-5}{*}{\raggedright\arraybackslash (30,30)} & 2000 & 0.0820 & 0.0820 & 0.0270 & 0.1015 & 0.0480 & 0.0565 & 0.0980 & 0.9480\\
			\cmidrule{1-10}
			& 100 & 0.0755 & 0.0760 & 0.0545 & 0.0600 & 0.0455 & 0.0700 & 0.0935 & 0.9695\\
			& 200 & 0.0775 & 0.0770 & 0.0505 & 0.0640 & 0.0480 & 0.0710 & 0.0895 & 0.9760\\
			& 500 & 0.0675 & 0.0675 & 0.0235 & 0.0625 & 0.0455 & 0.0600 & 0.0735 & 0.9830\\
			& 1000 & 0.0745 & 0.0725 & 0.0240 & 0.0630 & 0.0505 & 0.0615 & 0.0710 & 0.9880\\
			\multirow{-5}{*}{\raggedright\arraybackslash (30,50)} & 2000 & 0.0760 & 0.0760 & 0.0185 & 0.0720 & 0.0440 & 0.0600 & 0.0685 & 0.9880\\
			\cmidrule{1-10}
			& 100 & 0.0785 & 0.0785 & 0.0525 & 0.0535 & 0.0580 & 0.0800 & 0.0915 & 0.9820\\
			& 200 & 0.0685 & 0.0685 & 0.0415 & 0.0505 & 0.0455 & 0.0715 & 0.0785 & 0.9910\\
			& 500 & 0.0640 & 0.0640 & 0.0275 & 0.0450 & 0.0500 & 0.0575 & 0.0630 & 0.9930\\
			& 1000 & 0.0710 & 0.0710 & 0.0235 & 0.0605 & 0.0565 & 0.0565 & 0.0630 & 0.9920\\
			\multirow{-5}{*}{\raggedright\arraybackslash (50,50)} & 2000 & 0.0760 & 0.0760 & 0.0145 & 0.0575 & 0.0445 & 0.0535 & 0.0555 & 0.9970\\
			\cmidrule{1-10}
			& 100 & 0.0795 & 0.0780 & 0.0480 & 0.0510 & 0.0515 & 0.0670 & 0.0710 & 0.9950\\
			& 200 & 0.0715 & 0.0705 & 0.0395 & 0.0495 & 0.0510 & 0.0660 & 0.0690 & 0.9950\\
			& 500 & 0.0610 & 0.0610 & 0.0270 & 0.0365 & 0.0520 & 0.0595 & 0.0600 & 0.9990\\
			& 1000 & 0.0770 & 0.0780 & 0.0235 & 0.0500 & 0.0390 & 0.0555 & 0.0570 & 0.9980\\
			\multirow{-5}{*}{\raggedright\arraybackslash (50,100)} & 2000 & 0.0700 & 0.0690 & 0.0170 & 0.0535 & 0.0450 & 0.0485 & 0.0485 & 0.9995\\
			\cmidrule{1-10}
			& 100 & 0.0625 & 0.0625 & 0.0375 & 0.0450 & 0.0540 & 0.0535 & 0.0570 & 0.9965\\
			& 200 & 0.0740 & 0.0740 & 0.0340 & 0.0470 & 0.0445 & 0.0865 & 0.0875 & 0.9975\\
			& 500 & 0.0685 & 0.0685 & 0.0285 & 0.0425 & 0.0495 & 0.0615 & 0.0620 & 0.9990\\
			& 1000 & 0.0670 & 0.0670 & 0.0165 & 0.0410 & 0.0480 & 0.0545 & 0.0555 & 0.9990\\
			\multirow{-5}{*}{\raggedright\arraybackslash (100,100)} & 2000 & 0.0730 & 0.0730 & 0.0120 & 0.0430 & 0.0535 & 0.0520 & 0.0525 & 0.9995\\
			\bottomrule
	\end{tabular}}
\end{table}

\begin{table}[htp]
	\centering
	\caption{\label{table:sizemodel2}The sizes of the seven tests for Case 2 under short‑term dependence structure at significance level $\alpha = 0.05$.  The last column reports the simulated $\mathbb{P}(J_0 = 0)$. }
	\centering
	\resizebox{\ifdim\width>\linewidth\linewidth\else\width\fi}{!}{
		\fontsize{10}{12}\selectfont
		\begin{tabular}[t]{lrlllllllc}
			\toprule
			($n_1$, $n_2$) & $p$ & BS & CQ & SKK & CLX & LWJ & CRPT & CRPT-PE & $\mathbb{P}(J_0 = 0)$\\
			\midrule
			& 100 & 0.0605 & 0.0640 & 0.1400 & 0.1600 & 0.0540 & 0.0690 & 0.2080 & 0.4105\\
			& 200 & 0.0570 & 0.0575 & 0.1800 & 0.2145 & 0.0535 & 0.0585 & 0.1855 & 0.3030\\
			& 500 & 0.0540 & 0.0530 & 0.2640 & 0.2945 & 0.0485 & 0.0485 & 0.1835 & 0.2120\\
			& 1000 & 0.0505 & 0.0500 & 0.3970 & 0.3580 & 0.0445 & 0.0565 & 0.1805 & 0.1140\\
			\multirow{-9}{*}{\raggedright\arraybackslash (10,15)} & 2000 & 0.0530 & 0.0525 & 0.6105 & 0.4655 & 0.0450 & 0.0535 & 0.1790 & 0.0520\\
			\cmidrule{1-10}
			& 100 & 0.0650 & 0.0650 & 0.1060 & 0.1195 & 0.0525 & 0.0735 & 0.2125 & 0.5800\\
			& 200 & 0.0570 & 0.0570 & 0.1290 & 0.1620 & 0.0530 & 0.0575 & 0.2010 & 0.5305\\
			& 500 & 0.0515 & 0.0515 & 0.1745 & 0.2200 & 0.0440 & 0.0420 & 0.1880 & 0.4205\\
			& 1000 & 0.0595 & 0.0595 & 0.2560 & 0.2945 & 0.0470 & 0.0660 & 0.2065 & 0.3655\\
			\multirow{-9}{*}{\raggedright\arraybackslash (15,15)} & 2000 & 0.0490 & 0.0490 & 0.3705 & 0.3645 & 0.0520 & 0.0510 & 0.2100 & 0.2725\\
			\cmidrule{1-10}
			& 100 & 0.0565 & 0.0590 & 0.0960 & 0.0875 & 0.0550 & 0.0735 & 0.1635 & 0.8515\\
			& 200 & 0.0505 & 0.0495 & 0.1065 & 0.1060 & 0.0565 & 0.0520 & 0.1445 & 0.8465\\
			& 500 & 0.0480 & 0.0525 & 0.1570 & 0.1380 & 0.0490 & 0.0585 & 0.1505 & 0.8385\\
			& 1000 & 0.0605 & 0.0565 & 0.2345 & 0.1765 & 0.0495 & 0.0515 & 0.1605 & 0.8235\\
			\multirow{-5}{*}{\raggedright\arraybackslash (15,30)} & 2000 & 0.0540 & 0.0575 & 0.3430 & 0.2285 & 0.0510 & 0.0470 & 0.1650 & 0.7970\\
			\cmidrule{1-10}
			& 100 & 0.0525 & 0.0525 & 0.0630 & 0.0670 & 0.0530 & 0.0735 & 0.1230 & 0.9355\\
			& 200 & 0.0565 & 0.0565 & 0.0765 & 0.0940 & 0.0540 & 0.0640 & 0.1135 & 0.9300\\
			& 500 & 0.0555 & 0.0555 & 0.0975 & 0.1010 & 0.0490 & 0.0555 & 0.0975 & 0.9430\\
			& 1000 & 0.0525 & 0.0525 & 0.1060 & 0.1130 & 0.0485 & 0.0635 & 0.0960 & 0.9535\\
			\multirow{-5}{*}{\raggedright\arraybackslash (30,30)} & 2000 & 0.0505 & 0.0505 & 0.1335 & 0.1600 & 0.0475 & 0.0530 & 0.0985 & 0.9430\\
			\cmidrule{1-10}
			& 100 & 0.0550 & 0.0555 & 0.0685 & 0.0565 & 0.0435 & 0.0725 & 0.0935 & 0.9730\\
			& 200 & 0.0475 & 0.0475 & 0.0685 & 0.0775 & 0.0480 & 0.0665 & 0.0900 & 0.9715\\
			& 500 & 0.0515 & 0.0535 & 0.0895 & 0.0865 & 0.0570 & 0.0535 & 0.0735 & 0.9765\\
			& 1000 & 0.0505 & 0.0490 & 0.0980 & 0.0950 & 0.0435 & 0.0450 & 0.0605 & 0.9800\\
			\multirow{-5}{*}{\raggedright\arraybackslash (30,50)} & 2000 & 0.0495 & 0.0495 & 0.1335 & 0.1175 & 0.0515 & 0.0445 & 0.0600 & 0.9825\\
			\cmidrule{1-10}
			& 100 & 0.0635 & 0.0635 & 0.0720 & 0.0550 & 0.0470 & 0.0745 & 0.0845 & 0.9855\\
			& 200 & 0.0640 & 0.0640 & 0.0740 & 0.0675 & 0.0540 & 0.0745 & 0.0850 & 0.9870\\
			& 500 & 0.0550 & 0.0550 & 0.0725 & 0.0635 & 0.0445 & 0.0585 & 0.0640 & 0.9915\\
			& 1000 & 0.0535 & 0.0535 & 0.0790 & 0.0855 & 0.0510 & 0.0535 & 0.0625 & 0.9895\\
			\multirow{-5}{*}{\raggedright\arraybackslash (50,50)} & 2000 & 0.0550 & 0.0550 & 0.0980 & 0.1045 & 0.0440 & 0.0500 & 0.0555 & 0.9945\\
			\cmidrule{1-10}
			& 100 & 0.0670 & 0.0650 & 0.0670 & 0.0555 & 0.0490 & 0.0670 & 0.0735 & 0.9925\\
			& 200 & 0.0560 & 0.0565 & 0.0650 & 0.0575 & 0.0455 & 0.0815 & 0.0860 & 0.9945\\
			& 500 & 0.0520 & 0.0525 & 0.0705 & 0.0650 & 0.0555 & 0.0710 & 0.0720 & 0.9985\\
			& 1000 & 0.0505 & 0.0500 & 0.0735 & 0.0745 & 0.0475 & 0.0555 & 0.0560 & 0.9995\\
			\multirow{-5}{*}{\raggedright\arraybackslash (50,100)} & 2000 & 0.0540 & 0.0525 & 0.0960 & 0.0810 & 0.0535 & 0.0525 & 0.0560 & 0.9965\\
			\cmidrule{1-10}
			& 100 & 0.0595 & 0.0595 & 0.0570 & 0.0565 & 0.0505 & 0.0545 & 0.0570 & 0.9975\\
			& 200 & 0.0490 & 0.0490 & 0.0495 & 0.0490 & 0.0515 & 0.0800 & 0.0810 & 0.9990\\
			& 500 & 0.0495 & 0.0495 & 0.0575 & 0.0560 & 0.0425 & 0.0675 & 0.0675 & 1.0000\\
			& 1000 & 0.0455 & 0.0455 & 0.0560 & 0.0695 & 0.0520 & 0.0640 & 0.0650 & 0.9990\\
			\multirow{-5}{*}{\raggedright\arraybackslash (100,100)} & 2000 & 0.0465 & 0.0465 & 0.0630 & 0.0600 & 0.0465 & 0.0525 & 0.0530 & 0.9995\\
			\bottomrule
	\end{tabular}}
\end{table}

\begin{table}[htp]
	\centering
	\caption{\label{table:sizemodel3}The sizes of the seven tests for Case 3 under short‑term dependence structure at significance level $\alpha = 0.05$.  The last column reports the simulated $\mathbb{P}(J_0 = 0)$. }
	\centering
	\resizebox{\ifdim\width>\linewidth\linewidth\else\width\fi}{!}{
		\fontsize{10}{12}\selectfont
		\begin{tabular}[t]{lrlllllllc}
			\toprule
			($n_1$, $n_2$) & $p$ & BS & CQ & SKK & CLX & LWJ & CRPT & CRPT-PE & $\mathbb{P}(J_0 = 0)$\\
			\midrule
			& 100 & 0.0475 & 0.0460 & 0.0310 & 0.0695 & 0.0435 & 0.0360 & 0.1300 & 0.6330\\
			& 200 & 0.0445 & 0.0435 & 0.0250 & 0.0925 & 0.0435 & 0.0345 & 0.1190 & 0.6080\\
			& 500 & 0.0490 & 0.0435 & 0.0215 & 0.1175 & 0.0355 & 0.0310 & 0.1260 & 0.5785\\
			& 1000 & 0.0520 & 0.0470 & 0.0125 & 0.1295 & 0.0395 & 0.0300 & 0.1195 & 0.5700\\
			\multirow{-5}{*}{\raggedright\arraybackslash (10,15)} & 2000 & 0.0495 & 0.0460 & 0.0150 & 0.1415 & 0.0360 & 0.0395 & 0.1225 & 0.5365\\
			\cmidrule{1-10}
			& 100 & 0.0450 & 0.0450 & 0.0270 & 0.0560 & 0.0335 & 0.0320 & 0.1245 & 0.7725\\
			& 200 & 0.0510 & 0.0510 & 0.0245 & 0.0795 & 0.0425 & 0.0405 & 0.1265 & 0.7475\\
			& 500 & 0.0500 & 0.0500 & 0.0155 & 0.0870 & 0.0390 & 0.0270 & 0.1035 & 0.7435\\
			& 1000 & 0.0505 & 0.0505 & 0.0140 & 0.0935 & 0.0410 & 0.0325 & 0.1115 & 0.7575\\
			\multirow{-5}{*}{\raggedright\arraybackslash (15,15)} & 2000 & 0.0460 & 0.0460 & 0.0120 & 0.1095 & 0.0370 & 0.0350 & 0.1145 & 0.7270\\
			\cmidrule{1-10}
			& 100 & 0.0555 & 0.0470 & 0.0275 & 0.0470 & 0.0365 & 0.0515 & 0.1160 & 0.9135\\
			& 200 & 0.0565 & 0.0495 & 0.0230 & 0.0550 & 0.0415 & 0.0585 & 0.1210 & 0.9135\\
			& 500 & 0.0530 & 0.0455 & 0.0175 & 0.0570 & 0.0450 & 0.0400 & 0.0880 & 0.9345\\
			& 1000 & 0.0610 & 0.0525 & 0.0170 & 0.0630 & 0.0385 & 0.0380 & 0.0850 & 0.9370\\
			\multirow{-5}{*}{\raggedright\arraybackslash (15,30)} & 2000 & 0.0660 & 0.0545 & 0.0145 & 0.0665 & 0.0505 & 0.0405 & 0.0865 & 0.9430\\
			\cmidrule{1-10}
			& 100 & 0.0510 & 0.0510 & 0.0280 & 0.0380 & 0.0355 & 0.0345 & 0.0665 & 0.9640\\
			& 200 & 0.0500 & 0.0500 & 0.0275 & 0.0375 & 0.0385 & 0.0340 & 0.0595 & 0.9715\\
			& 500 & 0.0550 & 0.0550 & 0.0220 & 0.0400 & 0.0265 & 0.0315 & 0.0500 & 0.9780\\
			\multirow{-4}{*}{\raggedright\arraybackslash (30,30)} & 1000 & 0.0430 & 0.0430 & 0.0155 & 0.0405 & 0.0385 & 0.0290 & 0.0455 & 0.9810\\
			\cmidrule{1-10}
			& 100 & 0.0645 & 0.0615 & 0.0370 & 0.0405 & 0.0405 & 0.0520 & 0.0685 & 0.9805\\
			& 200 & 0.0575 & 0.0565 & 0.0255 & 0.0350 & 0.0425 & 0.0475 & 0.0575 & 0.9870\\
			& 500 & 0.0555 & 0.0535 & 0.0180 & 0.0475 & 0.0410 & 0.0380 & 0.0440 & 0.9925\\
			\multirow{-4}{*}{\raggedright\arraybackslash (30,50)} & 1000 & 0.0605 & 0.0540 & 0.0135 & 0.0385 & 0.0445 & 0.0295 & 0.0360 & 0.9930\\
			\cmidrule{1-10}
			& 100 & 0.0435 & 0.0435 & 0.0250 & 0.0295 & 0.0370 & 0.0515 & 0.0610 & 0.9885\\
			& 200 & 0.0605 & 0.0605 & 0.0250 & 0.0370 & 0.0435 & 0.0340 & 0.0405 & 0.9930\\
			& 500 & 0.0505 & 0.0505 & 0.0130 & 0.0255 & 0.0405 & 0.0305 & 0.0340 & 0.9960\\
			\multirow{-4}{*}{\raggedright\arraybackslash (50,50)} & 1000 & 0.0535 & 0.0535 & 0.0135 & 0.0330 & 0.0370 & 0.0220 & 0.0250 & 0.9970\\
			\cmidrule{1-10}
			& 100 & 0.0530 & 0.0530 & 0.0265 & 0.0320 & 0.0375 & 0.0445 & 0.0470 & 0.9970\\
			& 200 & 0.0575 & 0.0535 & 0.0215 & 0.0380 & 0.0395 & 0.0525 & 0.0540 & 0.9980\\
			& 500 & 0.0585 & 0.0555 & 0.0175 & 0.0310 & 0.0435 & 0.0370 & 0.0375 & 0.9990\\
			\multirow{-4}{*}{\raggedright\arraybackslash (50,100)} & 1000 & 0.0550 & 0.0500 & 0.0125 & 0.0360 & 0.0425 & 0.0305 & 0.0310 & 0.9995\\
			\cmidrule{1-10}
			& 100 & 0.0620 & 0.0620 & 0.0290 & 0.0280 & 0.0370 & 0.0375 & 0.0400 & 0.9975\\
			& 200 & 0.0480 & 0.0480 & 0.0235 & 0.0215 & 0.0395 & 0.0395 & 0.0405 & 0.9990\\
			& 500 & 0.0585 & 0.0585 & 0.0180 & 0.0335 & 0.0360 & 0.0355 & 0.0360 & 0.9995\\
			\multirow{-4}{*}{\raggedright\arraybackslash (100,100)} & 1000 & 0.0585 & 0.0585 & 0.0105 & 0.0255 & 0.0420 & 0.0280 & 0.0285 & 0.9990\\
			\bottomrule
	\end{tabular}}
\end{table}

\subsection{Statistical power}

Under alternative hypothesis, we set $\boldsymbol{\mu}_1 = \mathbf{0}$ without loss of generality. Let $\boldsymbol{\mu}_2$ have $s = \lfloor \tau p \rfloor$ nonzero entries of equal magnitude $\delta_0$ placed at randomly selected coordinates. 
The parameter $\tau$ controls the signal sparsity. The dimension is fixed at $p=1000$ and the significance level is $\alpha=0.05$. Two regimes are considered. The first is a dense regime with $\tau = 0.2$, corresponding to $s=200$. The second is a sparse regime with $\tau = 0.005$, which for $p=1000$ gives $s=5$. Two sample size configurations are examined, namely small samples with $(n_1,n_2) = (15,15)$ and moderate samples with $(50,50)$.
For the small sample case $(n_1,n_2)=(15,15)$, we exclude CLX and CRPT‑PE from the power figure because both methods fail to control the Type I error as shown in Tables~\ref{table:sizemodel1}-\ref{table:sizemodel3}. 
For the dense setting with $\tau=0.2$, 
we vary $\delta_0$ from $0$ to $1$ when $(n_1,n_2)=(15,15)$ and from $0$ to $0.5$ when $(n_1,n_2)=(50,50)$. 
For the sparse setting with $\tau=0.005$, the range is $0$ to $4$ for $(15,15)$ and $0$ to $2$ for $(50,50)$. 



\begin{figure}[htp]
	\begin{center}
		\includegraphics[width=15cm]{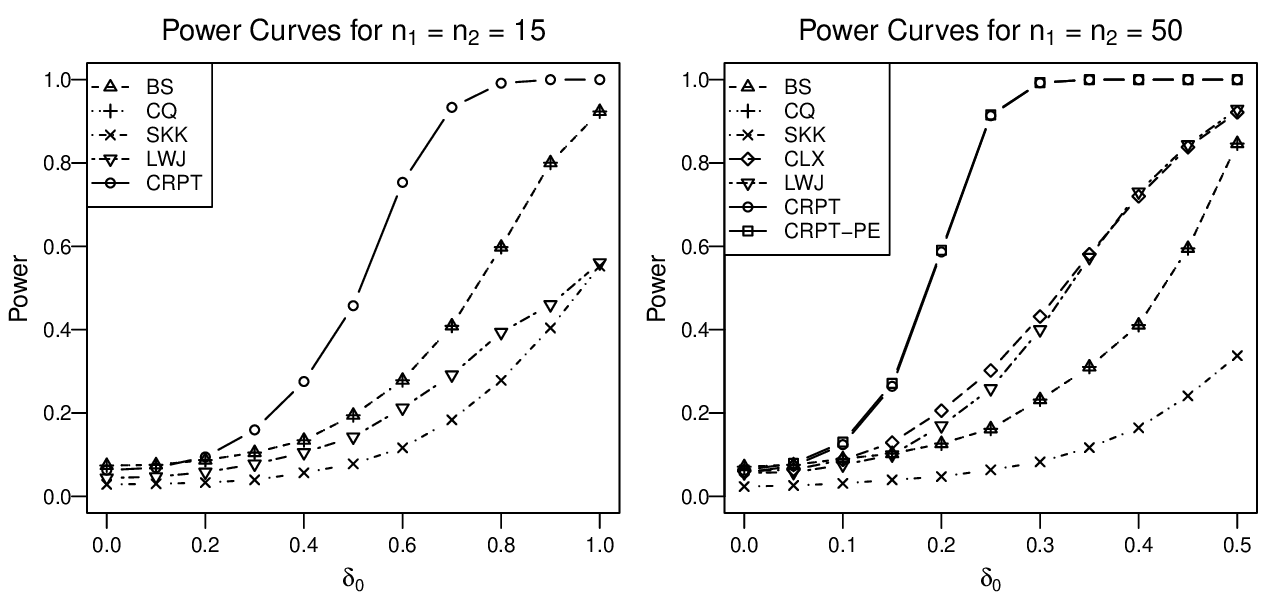}
		\caption{Power curves of the tests against the dense signals ($\tau$=0.2) under $p = 1000$ with $2000$ repetitions for Case 1. }\label{figure:powercase1_dense}
	\end{center}
\end{figure}

\begin{figure}[htp]
	\begin{center}
		\includegraphics[width=15cm]{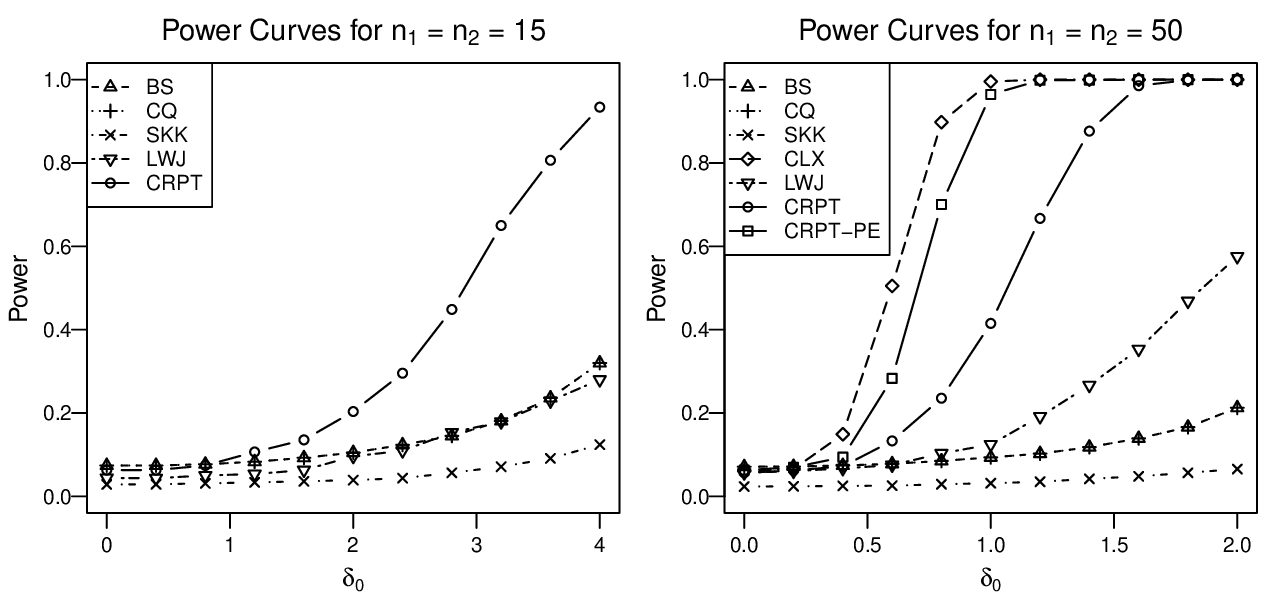}
		\caption{Power curves of the tests against the sparse signals ($\tau$=0.005) under $p = 1000$ with $2000$ repetitions for Case 1. }\label{figure:powercase1_sparse}
	\end{center}
\end{figure}

Figure~\ref{figure:powercase1_dense} presents the power curves of the methods against the dense signals for Case 1. 
For the small sample size \((n_1,n_2)=(15,15)\), the proposed CRPT exhibits the best performance among the competing methods.
When the sample size increases to \((n_1,n_2)=(50,50)\), both CRPT and CRPT-PE achieve high power, and their performances are very close in the dense setting.
Figure~\ref{figure:powercase1_sparse} reports the power curves under the sparse setting 
\((\tau=0.005)\). 
In the small sample setting, CRPT exhibits competitive power among the methods. 
For the larger sample size \((n_1,n_2)=(50,50)\), the CLX method achieves the best performance under sparse alternatives, as expected for a maximum-type test. 
Meanwhile, the proposed power-enhanced method CRPT-PE substantially improves the performance of CRPT, leading to a noticeable gain in power.

\begin{figure}[htp]
	\begin{center}
		\includegraphics[width=15cm]{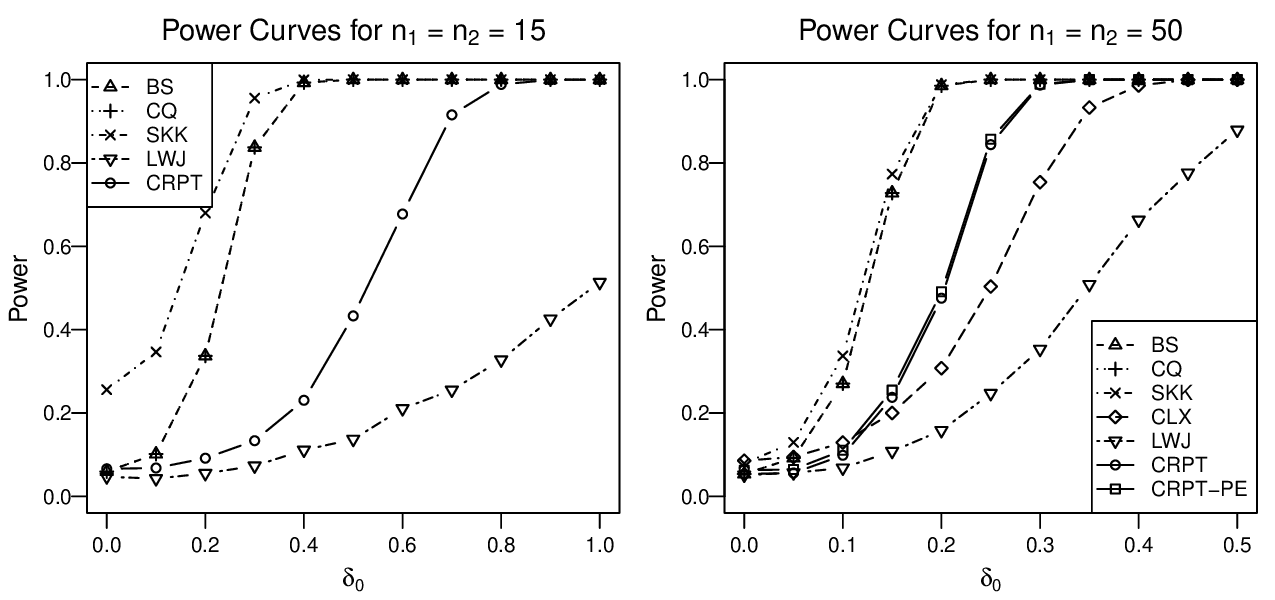}
		\caption{Power curves of the tests against the dense signals ($\tau$=0.2) under $p = 1000$ with $2000$ repetitions for Case 2. }\label{figure:powercase2_dense}
	\end{center}
\end{figure}

\begin{figure}[htp]
	\begin{center}
		\includegraphics[width=15cm]{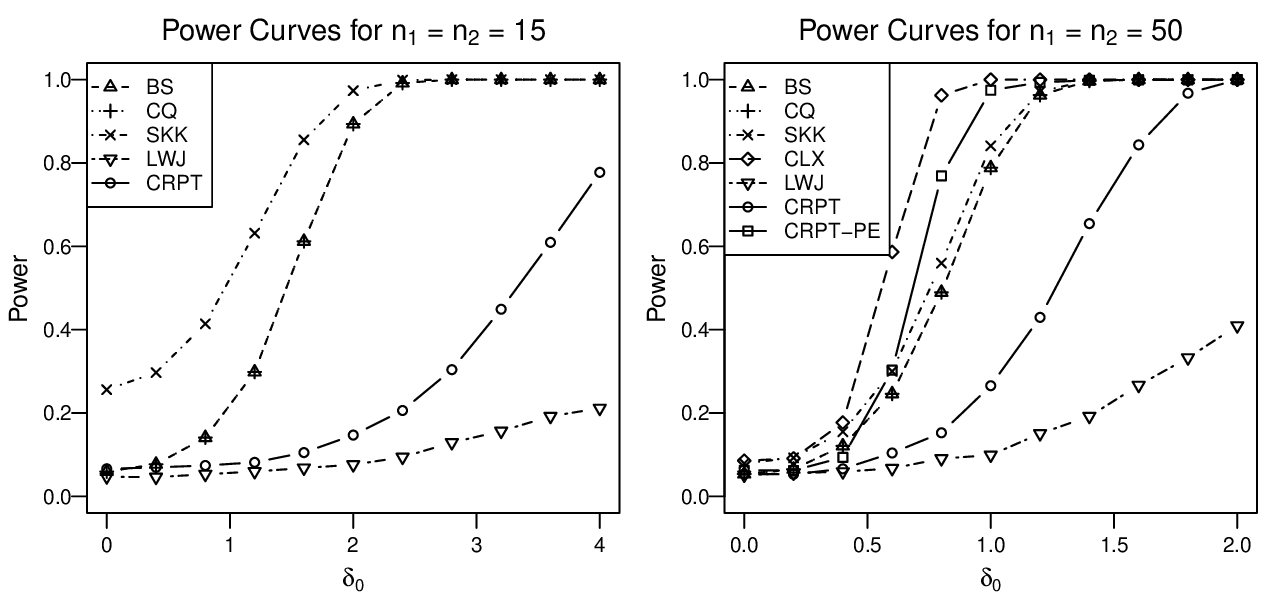}
		\caption{Power curves of the tests against the sparse signals ($\tau$=0.005) under $p = 1000$ with $2000$ repetitions for Case 2. }\label{figure:powercase2_sparse}
	\end{center}
\end{figure}

\begin{figure}[htp]
	\begin{center}
		\includegraphics[width=15cm]{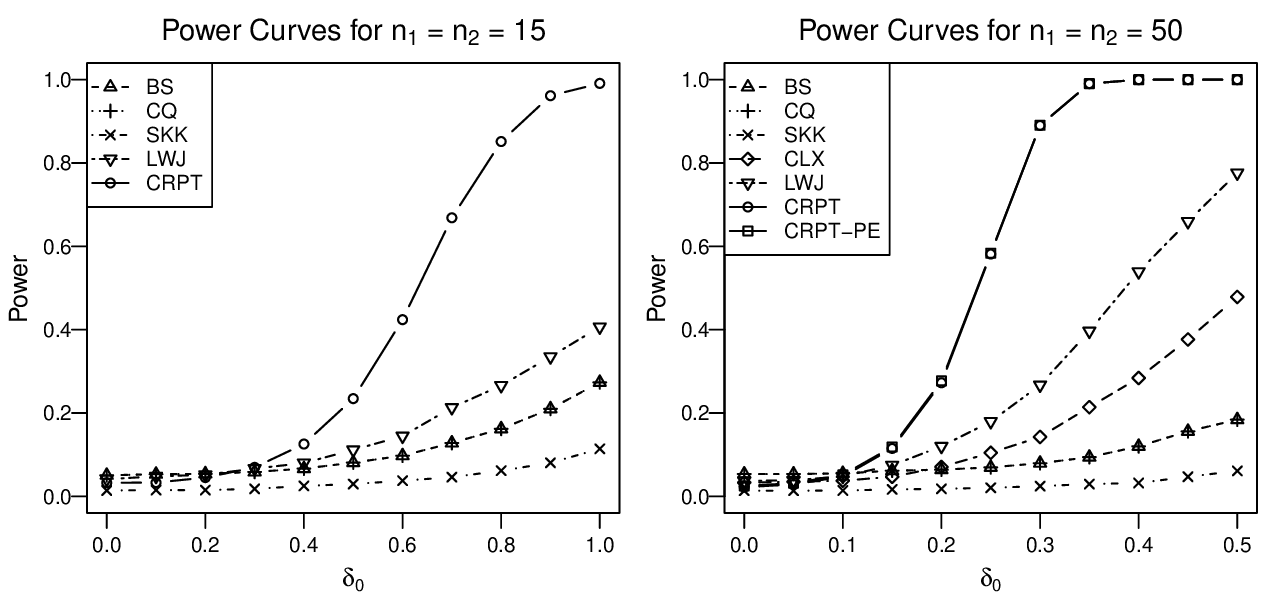}
		\caption{Power curves of the tests against the dense signals ($\tau$=0.2) under $p = 1000$ with $2000$ repetitions for Case 3. }\label{figure:powercase3_dense}
	\end{center}
\end{figure}

\begin{figure}[htp]
	\begin{center}
		\includegraphics[width=15cm]{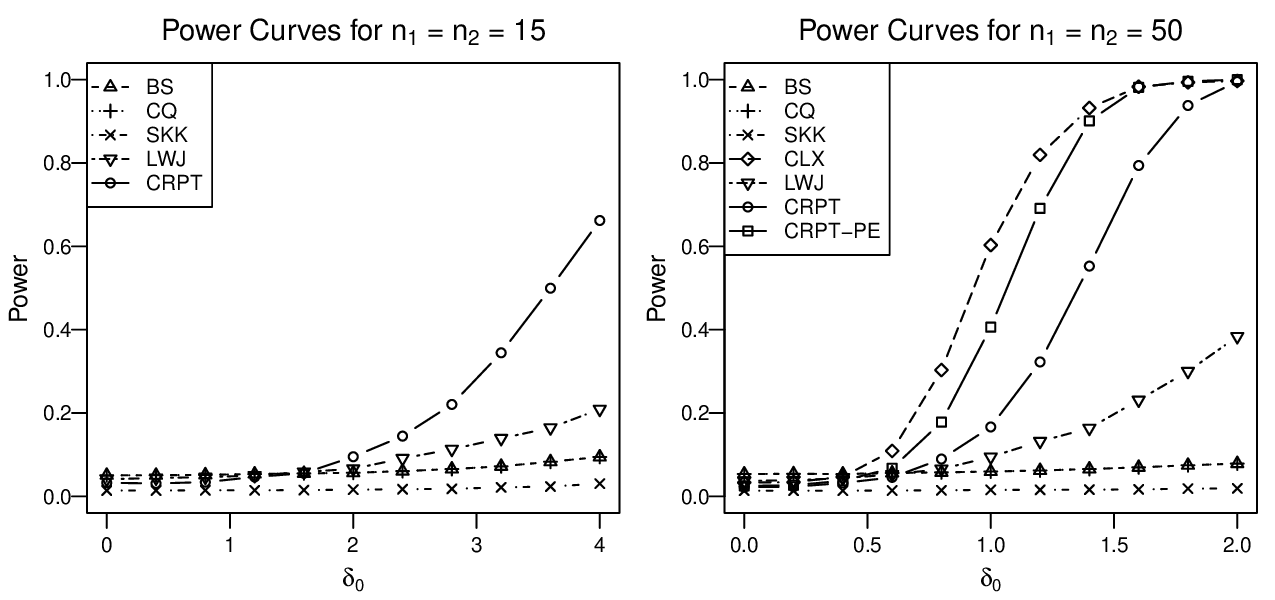}
		\caption{Power curves of the tests against the sparse signals ($\tau$=0.005) under $p = 1000$ with $2000$ repetitions for Case 3. }\label{figure:powercase3_sparse}
	\end{center}
\end{figure}

Figure~\ref{figure:powercase2_dense} presents the power curves of the methods against dense signals for Case~2 under the short-term dependence structure. 
For the small sample size \((n_1,n_2)=(15,15)\), SKK fails to control the Type~I error. 
Among the remaining methods, BS and CQ show strong performance, which is consistent with the fact that these two methods are particularly suitable when the covariance matrix is close to a diagonal or weakly dependent structure. 
The proposed CRPT shows a clear improvement over the single-projection method LWJ. 
Under \((n_1,n_2)=(50,50)\), the performances of CRPT and CRPT-PE are very close, while BS and CQ also remain highly competitive.
Figure~\ref{figure:powercase2_sparse} reports the power curves under the sparse setting \((\tau=0.005)\), where the pattern is similar to that observed in the dense setting with \((n_1,n_2)=(15,15)\). 
For \((n_1,n_2)=(50,50)\), the proposed power-enhanced method CRPT-PE can outperform BS and CQ for sufficiently large shift  $\delta_0$.

Figures~\ref{figure:powercase3_dense}-\ref{figure:powercase3_sparse} present the power curves of the methods for Case~3, where the data are generated from a multivariate \(t\)-distribution with long-term dependence. 
The overall power patterns are similar to those observed in Case~1. 
In the dense setting, the proposed CRPT maintains competitive power compared with the existing methods. 
In the sparse setting, the power-enhanced version CRPT-PE provides a clear improvement over the  CRPT. 
Although the theoretical properties of the proposed procedures are established under the Gaussian assumption, the simulation results for Case~3 indicate that the proposed methods could be robust to the $t$ distribution. 

\section{Real Data Analysis}
\label{sec:realdata}

In this section, we evaluate the proposed methods on a breast cancer gene expression dataset from the Gene Expression Omnibus database (GSE19159), which contains 168 early-stage breast cancer patients, available at \url{https://www.ncbi.nlm.nih.gov/geo/query/acc.cgi?acc=GSE19159}.
According to the clinical outcome, the patients are divided into two prognostic groups: 111 patients without recurrence within five years after diagnosis, referred to as the good-prognosis group, and 57 patients who developed distant metastasis within four years, referred to as the poor-prognosis group. 
There are 2905 genomic features retained for analysis. 

To evaluate the empirical power of different methods on this real dataset, we repeatedly draw random subsamples from the two prognostic groups. 
Specifically, we randomly select \(n_1\) observations from the good-prognosis group and \(n_2\) observations from the poor-prognosis group, and then apply all competing two-sample mean tests to the selected samples. 
This procedure is repeated 1000 times.
Table~\ref{table:breast} presents the empirical power with the significance level \(\alpha=0.05\).
For the small sample \((n_1,n_2)=(20,15)\), CLX and CRPT-PE achieve the highest empirical powers. 
However, as indicated by the simulation studies, these two methods may suffer from size distortion in small samples. 
Among the remaining methods, the proposed CRPT achieves the best performance. 
For the larger sample size \((n_1,n_2)=(50,30)\), both CRPT and CRPT-PE achieve the highest empirical powers. 
Overall, the proposed Cauchy-combined random projection methods are useful for identifying distributional differences in high-dimensional genomic studies.

\begin{table}[ht]
	\centering	\caption{Empirical powers of the seven methods for breast cancer prognosis data.}\label{table:breast}
	\begin{tabular}{lllllllll}
		\toprule
		 $n_1$ & $n_2$  & BS & CQ & SKK & CLX & LWJ & CRPT & CRPT-PE  \\ 
		 \hline
		 20 & 15 & 0.546 & 0.439 & 0.051 & 0.589 & 0.216 & 0.550 & 0.757 \\ 
		 50 & 30 & 0.995 & 0.973 & 0.604 & 0.948 & 0.622 & 0.996 & 0.996 \\ 
		\bottomrule
	\end{tabular}
\end{table}

%

\section{Conclusion}
\label{sec:conclusion}

In this paper, we propose a Cauchy-combined random projection test for high-dimensional two-sample mean problems. 
The proposed method applies Hotelling's \(T^2\) test after multiple random projections and then aggregates the resulting projected \(p\)-values through the Cauchy combination. 
We establish the asymptotic tail behavior of the proposed statistic under null hypothesis and further study its power consistency under suitable alternatives.
To improve the sensitivity of the proposed method to sparse mean differences, we also introduce a power-enhanced version of the test. 
The enhancement component enables the resulting statistic to retain the dense-signal detection ability of CRPT while gaining additional power for sparse signals. 
Simulation studies show that the proposed methods have competitive power under long-term dependence structure. 

Several extensions are worth future investigation. 
First, although the theoretical results are developed under the Gaussian assumption, it would be meaningful to extend the theory to broader distribution classes, such as sub-Gaussian or heavy-tailed distributions. 
Second, the current framework assumes a common covariance matrix for the two populations, and extending the method to unequal covariance settings is an important direction.



\bibliographystyle{plainnat} 
\bibliography{reference}

\newpage

\appendix

\begin{center}
 {\Large\bf  Appendix for ``Random Projection Tests via Cauchy Combination for  Two-Sample Mean''}   
\end{center}

This appendix provides detailed proofs of the main theoretical results and the supporting technical lemmas. Specifically, Appendix \ref{app:a} collects the notation used throughout the appendix. Appendix \ref{app:b} contains all technical lemmas and proofs of the main theorems. 

\section{Notations}\label{app:a}
We first recall some notations defined in the main text and introduce additional notations that will be used throughout this appendix. 
For $x_0 \in \mathbb{R} \cup \{\infty\}$, we write $A_x \sim B_x$ as $x \to x_0$
if $A_x/B_x \to 1$ as $x \to x_0$, while $A_x = o(B_x)$ as $x \to x_0$ represents $A_x/B_x \to 0$ as $x \to x_0$.
Moreover, $A_x \asymp B_x$ as $x \to x_0$ means that there exist constants
$0 < C_1 \le C_2 < \infty$ and a neighborhood of $x_0$ such that
$C_1 \le A_x/B_x \le C_2$
for all $x$ in that neighborhood.
Finally, $A_x = O(B_x)$ as $x \to x_0$ indicates that there exists a constant
$c > 0$ and a neighborhood of $x_0$ such that $|A_x| \le c |B_x|$
for all $x$ in that neighborhood.
When $x_0 = \infty$, the above conditions are understood to hold
for all sufficiently large $x$.
Let $\{Z_n\}$ be a sequence of random variables and $\{a_n\}$ be a sequence of positive real numbers.
We say that $Z_n = o_p(a_n)$ if $Z_n/a_n \stackrel{p}{\to} 0$ as $n\to\infty$. 
Let $Z_n = O_p(a_n)$ represent that  $Z_n/a_n$ is bounded in probability, that is, there exist constants $M>0$ such that $\lim_{n\to\infty}\mathbb{P}\left( |Z_n| \le M a_n \right) \to 1$.
Meanwhile, we say that $Z_n = \Omega_p(a_n)$ if there exist constants 
$M>0$ such that 
$\lim_{n\to\infty}\mathbb{P}\left( |Z_n| \ge M a_n \right) \to 1$.
Finally, we say that $Z_n = \Theta_p(a_n)$  if $Z_n = O_p(a_n)$ and $Z_n = \Omega_p(a_n)$.

Suppose that $F(m_1,m_2)$ denotes $F$-distribution with degrees of freedom $m_1$ and $m_2$ and $F_{m_1,m_2}(x)$ represents the cumulative distribution function of $F(m_1,m_2)$. 
Let $F\left(m_1,m_2,\gamma\right)$ denote the  noncentral $F$-distribution with degrees of freedom $m_1$ and $m_2$ and non-central parameter $\gamma$.
Let $\lambda_{\min}(\mathbf W)$ and $\lambda_{\max}(\mathbf W)$ denote the minimum and maximum eigenvalue of $\mathbf W$, respectively. 
Meanwhile, $\lambda_{\min}^+(\mathbf W)$ represents the minimum positive eigenvalue of $\mathbf W$. 
For two symmetric matrices $\mathbf A$ and $\mathbf B$ of the same dimension, 
we write $\mathbf A \preceq \mathbf B$ if $\mathbf B-\mathbf A$ is positive semi-definite. 
Similarly, $\mathbf A \prec \mathbf B$ means that $\mathbf B-\mathbf A$ is positive definite.

Suppose $\{\mathbf{x}_1, \ldots, \mathbf{x}_{n_1}\}$ and $\{\mathbf{y}_1, \ldots, \mathbf{y}_{n_2}\}$ are two $p$-dimensional independent samples.
Their sample means are $\bar{\mathbf{x}}$ and $\bar{\mathbf{y}}$, and $\mathbf{S}$ is the pooled sample covariance matrix.
For random projection matrices $\mathbf{P}_{b}$, the projection test statistics is defined as
\begin{align}\label{eq:t_statistics}
T_{b}^2 = \frac{n_1 n_2}{n_1+n_2} \big[\mathbf{P}_{b} (\bar{\mathbf{x}} - \bar{\mathbf{y}})\big]^\top (\mathbf{P}_{b} \mathbf{S} \mathbf{P}_{b}^\top)^{-1} \big[\mathbf{P}_{b} (\bar{\mathbf{x}} - \bar{\mathbf{y}})\big], \quad b=1,\ldots,B.
\end{align}
Let $\phi(p) := \tan\{(0.5-p)\pi\}$. 
The combined test statistics is defined as $T_{\text{CRPT}} =\sum_{b=1}^B\phi(p_b)/B$,
where $p_b = 1-F_{k,n-k+1}((n-k+1)T_{b}^2/(nk))$ and $n = n_1+n_2-2$.

\section{Technical details}\label{app:b}

\subsection{Proof of Theorem \ref{theo:t_tail}}

\subsubsection{Proof of Theorem \ref{theo:t_tail}}

\begin{proof}
	Let $\mathcal D=
	\{\mathbf x_1,\ldots,\mathbf x_{n_1},\mathbf y_1,\ldots,\mathbf y_{n_2}\}$
	denote the observed data. Since the projection matrices
	$\mathbf P_{i}$ and $\mathbf P_{j}$ are independent and are also
	independent of $\mathcal D$, the statistics $T_{i}^2$ and
	$T_{j}^2$ are conditionally independent given $\mathcal D$. Therefore,
	\[
	\begin{aligned}
		\mathbb P(T_{i}^2>t,T_{j}^2>t)
		&=
		\mathbb E\left[
		\mathbb P(T_{i}^2>t,T_{j}^2>t\mid\mathcal D)
		\right]  \\
		&=
		\mathbb E\left[
		\mathbb P(T_{i}^2>t\mid\mathcal D)
		\mathbb P(T_{j}^2>t\mid\mathcal D)
		\right].
	\end{aligned}
	\]
	Since $\mathbf P_{i}$ and $\mathbf P_{j}$ have the same distribution,
	we may write $\pi_t(\mathcal D)
	=\mathbb P(T_{b}^2>t\mid\mathcal D)$,
	where $b$ is any fixed projection index. Then
	\begin{align}\label{eq:theo1_titj}
	\mathbb P(T_{i}^2>t,T_{j}^2>t)=\mathbb E\{\pi_t(\mathcal D)^2\}.
	\end{align}

	Let $\mathcal G_{p,n}$ be the event defined in
	Lemma~\ref{lemm:hd_good}. By Lemma~\ref{lemm:hd_good}, when $p\ge 4n$, we have
	\begin{align}\label{eq:gc}
	\mathbb P(\mathcal G_{p,n}^c)
	\le 4\exp\{-p/32\}.
	\end{align}
	Let $0<\beta<1/2$. By Lemma~\ref{lemm:cond_tail_hd} (i), as $t\to\infty$, we have
	\begin{align}\label{eq:pid}
	\pi_t(\mathcal D)
	= O\big[\exp\{-t^\beta/8\} + t^{-(1-\beta){(n-k+1)/2}} \big],
	\qquad \mathcal D\in\mathcal G_{p,n}.
	\end{align}
	Combining \eqref{eq:gc}-\eqref{eq:pid}, we have
	\begin{align*}
		\mathbb E\{\pi_t(\mathcal D)^2\}
		&=
		\mathbb E\{\pi_t(\mathcal D)^2\mathbf 1_{\mathcal G_{p,n}}\}
		+
		\mathbb E\{\pi_t(\mathcal D)^2\mathbf 1_{\mathcal G_{p,n}^c}\} \le
		\sup_{\mathcal D\in\mathcal G_{p,n}}\pi_t(\mathcal D)^2
		+
		\mathbb P(\mathcal G_{p,n}^c)\\
		&=	O\big[\exp\{-t^\beta/4\}
		+ t^{-2(1-\beta){(n-k+1)/2}} + \exp\{-p/32\} \big]
	\end{align*}
	Since $0<\beta<1/2$, with $p/\log t\to\infty$, we have
	\begin{align}\label{eq:theo_pid}
	\mathbb E\{\pi_t(\mathcal D)^2\}=o(t^{-(n-k+1)/2}).
	\end{align}
	Meanwhile, by Lemma~\ref{lemm:t_tail}, we have
	$\mathbb P(T_{i}^2>t)
	\asymp t^{-(n-k+1)/2}$.
	With \eqref{eq:theo1_titj} and \eqref{eq:theo_pid}, we have
	\begin{align*}
	\mathbb P(T_{i}^2>t,T_{j}^2>t)
	= \mathbb E\{\pi_t(\mathcal D)^2\}
	= o(t^{-(n-k+1)/2})
	= o\{\mathbb P(T_{i}^2>t)\},
	\end{align*}
	which completes the proof.
\end{proof}

\subsubsection{Technical Lemma for Theorem \ref{theo:t_tail}}

\begin{lemm}\label{lemm:hd_good}
	Suppose Conditions \ref{cond:normal} and \ref{cond:eigenvalues} hold and $H_0$ is true. 
	Let $c=n_1n_2/(n_1+n_2)$ and
	$\mathbf d=\bar{\mathbf x}-\bar{\mathbf y}$.
	Assume $p\ge 4n$. Define the event
	\begin{align}
		\mathcal G_{p,n}
		=\bigg\{
		&\lambda_{\min}^+(\mathbf S)\ge \frac{\lambda_{\min}(\boldsymbol\Sigma) p}{16n},
		\quad
		\lambda_{\max}(\mathbf S)\le \frac{25\lambda_{\max}(\boldsymbol\Sigma) p}{4n},\notag \\
		&\frac{\lambda_{\min}(\boldsymbol\Sigma) p}{4}
		\le c\|\mathbf d\|^2 \le 5\lambda_{\max}(\boldsymbol\Sigma) p
		\bigg\}.
	\end{align}
	Then, it holds that
	\begin{align}
		\mathbb P(\mathcal G_{p,n}^c)
		\le 4e^{-p/32}.
	\end{align}
	Moreover, on $\mathcal G_{p,n}$, we have 
	\begin{align}
		\frac{c\|\mathbf d\|^2}{\lambda_{\min}^+(\mathbf S)}
		\le
		80\frac{\lambda_{\max}(\boldsymbol\Sigma)}{\lambda_{\min}(\boldsymbol\Sigma)} n
		\quad {\rm and} \quad
		\frac{c\|\mathbf d\|^2}{\lambda_{\max}(\mathbf S)}
		\ge
		\frac{\lambda_{\min}(\boldsymbol\Sigma)}{25\lambda_{\max}(\boldsymbol\Sigma)}n.
	\end{align}
\end{lemm}

\begin{proof}
	Under null hypothesis,
	$\mathbf d$ follows
	$N\left(\mathbf 0, \boldsymbol\Sigma/c\right)$ and then $c\|\boldsymbol\Sigma^{-1/2}\mathbf d\|^2$ follows the chi-square distribution with the degree $p$.
	From concentration bounds for Gaussian quadratic forms \citep[Lemma 1]{2000Laurent}, for any $u_1>0$ and $u_2>0$, it holds that
	$\mathbb P\{c\|\boldsymbol\Sigma^{-1/2}\mathbf d\|^2\ge p+2\sqrt{pu_1}+2u_1\}\le e^{-u_1}$ and 	$\mathbb P\{c\|\boldsymbol\Sigma^{-1/2}\mathbf d\|^2\le p-2\sqrt{pu_2}\}\le e^{-u_2}$.
	Taking $u_1 = p$ and $u_2 = 9p/64$, we have 
	\begin{align*}
		&\mathbb P\{c\|\boldsymbol\Sigma^{-1/2}\mathbf d\|^2\ge 5p\}\le e^{-p}, \\
		&\mathbb P\{c\|\boldsymbol\Sigma^{-1/2}\mathbf d\|^2\le p/4\}\le e^{-9p/64}.
	\end{align*}
	Since $\lambda_{\min}(\boldsymbol\Sigma)
	\|\boldsymbol\Sigma^{-1/2}\mathbf d\|^2
	\le \|\mathbf d\|^2
	\le \lambda_{\max}(\boldsymbol\Sigma)
	\|\boldsymbol\Sigma^{-1/2}\mathbf d\|^2$,
	we obtain
	\begin{align}\label{eq:lemm4_d}
		\mathbb P\left(
		\frac{\lambda_{\min}(\boldsymbol\Sigma) p}{4}
		\le c\|\mathbf d\|^2
		\le 5\lambda_{\max}(\boldsymbol\Sigma) p
		\right)
		\ge 1-e^{-p}-e^{-9p/64}.
	\end{align}
	
	Meanwhile, under Condition \ref{cond:normal}, 
	$n\mathbf S$ follows the Wishart distribution $W_p(\boldsymbol\Sigma,n)$.
	Hence we may write
	\begin{align*}
		\mathbf S
		=
		\frac1n
		\boldsymbol\Sigma^{1/2}
		\mathbf W^\top\mathbf W
		\boldsymbol\Sigma^{1/2},
	\end{align*}
	where $\mathbf W\in\mathbb R^{n\times p}$ has independent standard normal
	entries. Since $p\ge 4n$, the matrix $\mathbf S$ has rank $n$ almost surely. 
	The nonzero eigenvalues of $\boldsymbol\Sigma^{1/2}\mathbf W^\top\mathbf W\boldsymbol\Sigma^{1/2}$
	are equal to the eigenvalues of $\mathbf W\boldsymbol\Sigma\mathbf W^\top$.
	Therefore, it follows that
	\begin{align}\label{eq:lemm4_lamda}
		\lambda_{\min}^+(\mathbf S)
		=
		\frac1n
		\lambda_{\min}(\mathbf W\boldsymbol\Sigma\mathbf W^\top)
		\ge
		\frac{\lambda_{\min}(\boldsymbol\Sigma)}{n}
		\lambda_{\min}(\mathbf W\mathbf W^\top)
		=
		\frac{\lambda_{\min}(\boldsymbol\Sigma)}{n}s_{\min}^2(\mathbf W).
	\end{align}
	
	From Theorem 6.1 in \cite{2019wainwright}, for $p>n$ and every $u>0$, we have
	\begin{align*}
		\mathbb P\left\{s_{\min}(\mathbf W)\le\sqrt p-\sqrt n-u\right\}\le \exp\{-u^2/2\}.
	\end{align*}
	Taking $u = (\sqrt{p}-\sqrt{n})/2$, we get
	$\mathbb P\left\{s_{\min}(\mathbf W)
	\le(\sqrt p-\sqrt n)/2\right\}
	\le\exp\left\{-(\sqrt p-\sqrt n)^2/8\right\}$.
	Since $p\ge4n$ and then $\sqrt p-\sqrt n\ge \sqrt p/2$,
	we have $\mathbb P\left\{s_{\min}^2(\mathbf W)\le p/16 \right\}
	\le e^{-p/32}$. With \eqref{eq:lemm4_lamda}, it follows that
	\begin{align}\label{eq:lemm4_lambdamin}
		\mathbb{P}\left(\lambda_{\min}^+(\mathbf S)
		\ge
		\frac{\lambda_{\min}(\boldsymbol\Sigma) p}{16n}\right) \ge 1-e^{-p/32}.
	\end{align}
	
	Similarly, notice that
	\begin{align}\label{eq:lemm4_lmax}
		\lambda_{\max}(\mathbf S)
		\le
		\frac{\lambda_{\max}(\boldsymbol\Sigma)}{n}s_{\max}^2(\mathbf W).
	\end{align}
	From Theorem 6.1 in \cite{2019wainwright}, for $p>n$ and every $u>0$, we have
	\begin{align*}
		\mathbb P\left\{s_{\max}(\mathbf W)\ge\sqrt p+\sqrt n+u\right\}\le \exp\{-u^2/2\}.
	\end{align*}
	Since $p\ge4n$, we have $\sqrt n\le \sqrt p/2$. 
	Taking $u = \sqrt{p}$, it follows that
	$\mathbb P\left\{s_{\max}^2(\mathbf W)\ge 25p/4\right\}\le \exp\{-p/2\}$. 
	Therefore, with \eqref{eq:lemm4_lmax}, we have
	\begin{align}\label{eq:lemm4_lambdamax}
		\mathbb P\left\{\lambda_{\max}(\mathbf S)
		\le
		\frac{25\lambda_{\max}(\boldsymbol\Sigma) p}{4n}\right\}\ge  1- \exp\{-p/2\}.	
	\end{align}
	
	Combining \eqref{eq:lemm4_d}, \eqref{eq:lemm4_lambdamin} and \eqref{eq:lemm4_lambdamax}, we have
	\[
	\mathbb P(\mathcal G_{p,n}^c)
	\le e^{-p}+e^{-9p/64}+e^{-p/32} +e^{-p/2} \le 4e^{-p/32}.
	\]
	
	Finally, on $\mathcal G_{p,n}$,
	\[
	\frac{c\|\mathbf d\|^2}{\lambda_{\min}^+(\mathbf S)}
	\le
	\frac{5\lambda_{\max}(\boldsymbol\Sigma) p}{\lambda_{\min}(\boldsymbol\Sigma) p/(16n)}
	=
	80\frac{\lambda_{\max}(\boldsymbol\Sigma)}{\lambda_{\min}(\boldsymbol\Sigma)} n,
	\]
	and
	\[
	\frac{c\|\mathbf d\|^2}{\lambda_{\max}(\mathbf S)}
	\ge
	\frac{\lambda_{\min}(\boldsymbol\Sigma) p/4}{25\lambda_{\max}(\boldsymbol\Sigma) p/(4n)}
	=
	\frac{n\lambda_{\min}(\boldsymbol\Sigma)}{25\lambda_{\max}(\boldsymbol\Sigma)},
	\]
	which completes the proof.
\end{proof}

\begin{lemm}\label{lemm:cond_tail_hd}
	Suppose Conditions \ref{cond:normal}-\ref{cond:eigenvalues} hold and $H_0$ is true.
	Assume $p\ge4n$.
	Let $T_{b}^2$ be the statistics as \eqref{eq:t_statistics} corresponding to the random projection matrices $\mathbf{P}_{b}$, $b=1,\cdots,B$.
	For any $\mathcal D\in\mathcal G_{p,n}$ defined in Lemma \ref{lemm:hd_good}, the following bounds hold.
	
	(i) For any $t>0$ and $L>0$,
	\begin{align}
		\mathbb P(T_{b}^2>t\mid\mathcal D)
		\le
		\mathbb P(\chi_k^2>L)
		+
		\mathbb P\left\{
		\lambda_{\min}(\mathbf G\mathbf G^\top)
		<
		\frac{80\lambda_{\max}(\boldsymbol\Sigma) nL}{\lambda_{\min}(\boldsymbol\Sigma)t}\mid\mathcal D
		\right\},
	\end{align}
	where $\mathbf G\in\mathbb R^{k\times n}$ has independent standard normal entries.
	Consequently, as $L\to\infty$ and $nL/t\to0$, we have
	\begin{align}
		\mathbb P(T_{b}^2>t\mid\mathcal D)
		= O\left(\exp\{-L/8\}
		+ \left(
		\frac{80\lambda_{\max}(\boldsymbol\Sigma) nL}{\lambda_{\min}(\boldsymbol\Sigma)t}
		\right)^{(n-k+1)/2}\right).
	\end{align}

	(ii) For any $s>0$ and $0<h<1$,
	\begin{align}
		\mathbb P(T_{b}^2<s\mid\mathcal D)
		\le
		\mathbb P(\chi_k^2<h)
		+
		\mathbb P\left\{
		\lambda_{\max}(\mathbf G\mathbf G^\top)
		>
		\frac{\lambda_{\min}(\boldsymbol\Sigma)nh}{25\lambda_{\max}(\boldsymbol\Sigma) s}\mid\mathcal D
		\right\}.
	\end{align}
	Consequently, for fixed $n,k$, if $h/s\to\infty$ and $h
	\to 0^+$, then
	\begin{align}
		\mathbb P(T_{b}^2<s\mid\mathcal D)
		= O\left(h^{k/2}
		+ \exp\left\{-\frac{n\lambda_{\min}(\boldsymbol\Sigma)h}{100\lambda_{\max}(\boldsymbol\Sigma)s}\right\}\right)
	\end{align}
\end{lemm}

\begin{proof}

	Since $p>n$, $\operatorname{rank}(\mathbf S)=n$ almost surely.
	Hence, throughout the following arguments, we work on the event $\{\operatorname{rank}(\mathbf S)=n\}$, whose probability is one. 
	
	Rewrite the spectral decomposition of the pooled sample covariance matrix as $\mathbf S=\mathbf U_n\boldsymbol\Lambda_n\mathbf U_n^\top$,
	where $\mathbf U_n\in\mathbb R^{p\times n}$ has orthonormal columns and $\boldsymbol\Lambda_n
	= \operatorname{diag}(\lambda_1,\ldots,\lambda_n)$
	contains the positive eigenvalues of $\mathbf S$.
	Let $\mathbf G=\mathbf P_{b}\mathbf U_n$.
	Since $\mathbf P_{b}$ has independent standard normal entries and
	$\mathbf U_n$ has orthonormal columns, conditional on $\mathcal D$,
	\(\mathbf G\in\mathbb R^{k\times n}\) has independent standard normal entries.
	
	{\bf Proof of Lemma \eqref{lemm:cond_tail_hd} (i):}

	Notice that since $\mathbf P_{b}\mathbf S\mathbf P_{b}^\top
	= \mathbf G\boldsymbol\Lambda_n\mathbf G^\top$,
	we have
	\[
	\lambda_{\min}
	(\mathbf P_{b}\mathbf S\mathbf P_{b}^\top)
	\ge
	\lambda_{\min}^+(\mathbf S)
	\lambda_{\min}(\mathbf G\mathbf G^\top).
	\]
	Since $k<n$, the matrix $\mathbf G\mathbf G^\top$ is nonsingular almost surely and then
	\[
	T_{b}^2
	=c(\mathbf P_{b}\mathbf d)^\top
	(\mathbf P_{b}\mathbf S\mathbf P_{b}^\top)^{-1}
	(\mathbf P_{b}\mathbf d)
	\le
	\frac{
		c\|\mathbf P_{b}\mathbf d\|^2
	}{
		\lambda_{\min}^+(\mathbf S)
		\lambda_{\min}(\mathbf G\mathbf G^\top)
	}.
	\]
	
	Fixing $L>0$, on the event 
	$\{\|\mathbf P_{b}\mathbf d\|^2\le L\|\mathbf d\|^2, 
	t\lambda_{\min}(\mathbf G\mathbf G^\top)
	\ge cL\|\mathbf d\|^2/\lambda_{\min}^+(\mathbf S)\}$, we have $T_{b}^2\le t$.
	Hence, by the union bound,
	\begin{align}\label{eq:lemm2_t2}
		\mathbb P(T_{b}^2>t\mid\mathcal D)
		&\le
		\mathbb P\left(
		\|\mathbf P_{b}\mathbf d\|^2>L\|\mathbf d\|^2
		\mid\mathcal D
		\right)
		+ \mathbb P\left(
		\lambda_{\min}(\mathbf G\mathbf G^\top)
		<\frac{cL\|\mathbf d\|^2}
		{t\lambda_{\min}^+(\mathbf S)}
		\mid\mathcal D
		\right).
	\end{align}
	Conditional on $\mathcal D$, since $\mathbf P_{b}$ has independent standard normal entries, we have that $\|\mathbf P_{b}\mathbf d\|^2/\|\mathbf d\|^2$ follows the chi-square distribution with the degree $k$.
	Thus, it follows that
	\begin{align}\label{eq:lemm2_pd}
		\mathbb P\left(
		\|\mathbf P_{b}\mathbf d\|^2>L\|\mathbf d\|^2
		\mid\mathcal D
		\right)
		=\mathbb P(\chi_k^2>L).
	\end{align}
	Meanwhile, from Lemma 1 in \cite{2000Laurent}, for any $u>0$, we have $\mathbb{P}\big(\chi_k^2 > k + 2\sqrt{ku} + 2u\big) \le \exp\{-u\}$.
	Taking $u = L/8$, for sufficiently large $L$ satisfying $L>8k$, it follows that $k + 2\sqrt{ku} + 2u\le L$ and then $\mathbb P(\chi_k^2>L) \le \exp\{-L/8\}$. Therefore, as $L\to\infty$, we have
	\begin{align}
		P(\chi_k^2>L) = O(\exp\{-L/8\}).
	\end{align}

	On the event $\mathcal G_{p,n}$, Lemma \ref{lemm:hd_good} gives
	$c\|\mathbf d\|^2/\lambda_{\min}^+(\mathbf S) \le 80n\lambda_{\max}(\boldsymbol{\Sigma})/\lambda_{\min}(\boldsymbol{\Sigma})$.
	Consequently, on the event $\mathcal G_{p,n}$, we have
	\begin{align}
		\mathbb P\left(
		\lambda_{\min}(\mathbf G\mathbf G^\top)
		<\frac{cL\|\mathbf d\|^2}
		{t\lambda_{\min}^+(\mathbf S)}
		\mid\mathcal D
		\right) \le
		\mathbb P\left\{
		\lambda_{\min}(\mathbf G\mathbf G^\top)
		<
		\frac{80n\lambda_{\max}(\boldsymbol{\Sigma})L}{\lambda_{\min}(\boldsymbol{\Sigma})t}\mid\mathcal D
		\right\}.
	\end{align}
	Notice that $\mathbf G\mathbf G^\top$ follows the Wishart distribution $W_k(\mathbf I_k,n)$.
	From the probability density function of the minimum eigenvalue for the Wishart distribution \citep[Lemma 4.1]{Edelman1989Thesis}, as $u\to 0^+$, we have 
	$\mathbb P\{\lambda_{\min}(\mathbf G\mathbf G^\top)<u\mid\mathcal D\}=O(u^{(n-k+1)/2})$.
	Therefore, if $nL/t\to0$, it follows that
	\begin{align}\label{eq:lemm2_gg}
		\mathbb P\left\{
		\lambda_{\min}(\mathbf G\mathbf G^\top)
		<\frac{80n\lambda_{\max}(\boldsymbol{\Sigma})L}{\lambda_{\min}(\boldsymbol{\Sigma})t}\mid\mathcal D\right\}
		=O\left[\left(\frac{80n\lambda_{\max}(\boldsymbol{\Sigma})L}{\lambda_{\min}(\boldsymbol{\Sigma})t}\right)^{(n-k+1)/2}	\right].
	\end{align}
	Combining \eqref{eq:lemm2_t2}-\eqref{eq:lemm2_gg}, we finish the proof of (i).
	
	{\bf Proof of Lemma \eqref{lemm:cond_tail_hd} (ii):}
	
	Similarly, we have
	$\lambda_{\max}
	(\mathbf P_{b}\mathbf S\mathbf P_{b}^\top)
	\le
	\lambda_{\max}(\mathbf S)
	\lambda_{\max}(\mathbf G\mathbf G^\top)$ 
	and then
	\[
	T_{b}^2
	=
	c(\mathbf P_{b}\mathbf d)^\top
	(\mathbf P_{b}\mathbf S\mathbf P_{b}^\top)^{-1}
	(\mathbf P_{b}\mathbf d)
	\ge
	\frac{
		c\|\mathbf P_{b}\mathbf d\|^2
	}{
		\lambda_{\max}(\mathbf S)
		\lambda_{\max}(\mathbf G\mathbf G^\top)
	}.
	\]
	Fixing $0<h<1$, on the event 
	$\{\|\mathbf P_{b}\mathbf d\|^2\ge h\|\mathbf d\|^2, 
	\lambda_{\max}(\mathbf G\mathbf G^\top)
	\le
	c h\|\mathbf d\|^2/(s\lambda_{\max}(\mathbf S))\}$, we have $T_{b}^2\ge s$.
	Hence, by the union bound, it follows that
	\begin{align}\label{eq:lemm4_t2}
		\mathbb P(T_{b}^2<s\mid\mathcal D)
		&\le
		\mathbb P\left(
		\|\mathbf P_{b}\mathbf d\|^2<h\|\mathbf d\|^2
		\mid\mathcal D
		\right)+
		\mathbb P\left\{
		\lambda_{\max}(\mathbf G\mathbf G^\top)
		>
		\frac{c h\|\mathbf d\|^2}
		{s\lambda_{\max}(\mathbf S)}\mid\mathcal D
		\right\}  \notag \\
		&= \mathbb P(\chi_k^2<h)+
		\mathbb P\left\{
		\lambda_{\max}(\mathbf G\mathbf G^\top)
		>
		\frac{c h\|\mathbf d\|^2}
		{s\lambda_{\max}(\mathbf S)}\mid\mathcal D
		\right\}.
	\end{align}
	
	For fixed $k$, as $h\to 0^+$, we have
	\begin{align}\label{eq:lemm4_chi<h}
		\mathbb{P} (\chi_k^2<h)  = \frac{1}{2^{k/2}\Gamma(k/2)}\int_0^h u^{k/2-1}(1+o(1)){\rm d}u = O(h^{k/2}).
	\end{align}
	
	Recall that $\mathbf G$ has i.i.d $N(0,1)$ entries and $\mathbf G\mathbf G^\top$ follows the Wishart distribution $W_k(\mathbf I_k,n)$ conditional on $\mathcal D$.
	From Lemma 4.2 in \cite{1988Edelman}, for any $u>0$, there exists a constant $C_{n,k}$ such that 
	\begin{align*}
		\mathbb P\big(
		\lambda_{\max}(\mathbf G\mathbf G^\top)
		\ge u\mid\mathcal D
		\big) 
		\le C_{n,k}\int_{u}^\infty  v^{(n+k-3)/2}\exp\{-v/2\} {\rm d}v.
	\end{align*}
	On $\mathcal G_{p,n}$, we have
	$c\|\mathbf d\|^2/\lambda_{\max}(\mathbf S)
	\ge
	n\lambda_{\min}(\boldsymbol\Sigma)/\{25\lambda_{\max}(\boldsymbol\Sigma)\}$ and then
	\begin{align}\label{eq:lemm4_lmax>}
		\mathbb P\left\{
		\lambda_{\max}(\mathbf G\mathbf G^\top)
		>\frac{c h\|\mathbf d\|^2}
		{s\lambda_{\max}(\mathbf S)}\mid\mathcal D\right\}
		&\le \mathbb P\left\{
		\lambda_{\max}(\mathbf G\mathbf G^\top)
		>
		\frac{n\lambda_{\min}(\boldsymbol\Sigma)h}{25\lambda_{\max}(\boldsymbol\Sigma) s}\mid\mathcal D\right\} \notag\\
		&= O\left(\exp\left\{-\frac{n\lambda_{\min}(\boldsymbol\Sigma)h}{100\lambda_{\max}(\boldsymbol\Sigma)s}\right\}\right).
	\end{align}
	Combining \eqref{eq:lemm4_t2}-\eqref{eq:lemm4_lmax>}, we finish the proof of (ii).
\end{proof}

\begin{lemm}\label{lemm:t_tail}
	Let $\mathcal{D} = \{\mathbf{x}_1,\ldots,\mathbf{x}_{n_1},\mathbf{y}_1,\ldots,\mathbf{y}_{n_2}\}$ denote the observed data and $T_{b}^2$ be the statistics as \eqref{eq:t_statistics} corresponding to the random projection matrices $\mathbf{P}_{b}$, $b=1,\cdots,B$.
	Suppose the Conditions \ref{cond:normal}-\ref{cond:projection} and $H_0$ hold. 
	Let $c = n_1 n_2/(n_1+n_2)$, $\mathbf{d} = \bar{\mathbf{x}} - \bar{\mathbf{y}}$ and $t>0$. Then, as $t\to\infty$, 
	\begin{align}\label{eq:pt2}
		\mathbb P(T_{b}^2>t) \asymp t^{-(n-k+1)/2}. 
	\end{align}
\end{lemm}

\begin{proof} 
	
	Under null hypothesis and Condition (C1), $(n-k+1)T_{b}^2/(nk)$ conditioned on $\mathbf{P}_{b}$ follows $F(k,n-k+1)$. 
	Since $\mathbb{P}(T_{b}^2\le t) = \mathbb{E}[\mathbb{P}(T_{b}^2\le t \mid \mathbf{P}_{b})]$,
	the unconditional distribution of $T_{b}^2$ is also $F(k,n-k+1)$.
	
	Let $Z = (n-k+1)T_{b}^2/(nk)$, and then $Z \sim F(k, n-k+1)$. The probability density function of $Z$ is defined as $f_Z(\cdot)$ such that
	\[
	f_Z(z) = C \Big(\frac{k}{n-k+1}\Big)^{k/2} z^{k/2-1} \Big(1 + \frac{k}{n-k+1}z\Big)^{-\frac{n+1}{2}}, \quad z > 0, 
	\]
	where $C=\Gamma\left((n+1)/2\right)/\{\Gamma\left(k/2\right)\Gamma\left((n-k+1)/2\right)\}$.
	As $z \to \infty$, $f_Z(z) \asymp  z^{-(n-k+1)/2-1}$, and then
	\begin{align*}
		\mathbb{P}(Z>z) = \int_z^\infty f_Z(u) {\rm d}u \asymp z^{-(n-k+1)/2}.
	\end{align*}
	Therefore, $P(T_{b}^2>t) = \mathbb{P}(Z>(n-k+1)t/(nk))\asymp t^{-(n-k+1)/2}$, which completes the proof.
	
\end{proof}

\subsection{Proof of Theorem \ref{theo:null1}}

\subsubsection{Proof of Theorem \ref{theo:null1}}
\begin{proof} 
	
	Notice that $(n-k+1)T_{b}^2/(nk)$ follows $F(k,n-k+1)$ for $1\le b\le B$ under Condition (C1) and the null hypothesis. 
	Define
	\[
	U_b:=\phi(p_b)=\tan\{(0.5-p_b)\pi\},
	\qquad b=1,\ldots,B,
	\]
	where $p_b = 1-F_{k,n-k+1}((n-k+1)T_{b}^2/(nk))$ denotes the $p$-value corresponding to $T_{b}^2$. 
	Consequently, $U_b$ follows the standard Cauchy distribution under $H_0$ for $1\le b\le B$. 
	From Lemma \ref{lemm:u_quasi_hd}, $\{U_b, b=1,\ldots,B\}$ has pairwise asymptotic tail independence. Specifically, as $u\to\infty$ and $p/\log u\to\infty$, for $1\le i\not= j \le B$ and $a\in (0,1)$, we have 
	\begin{align}
		&\mathbb P(U_i>u,U_j>u^a) = o(\mathbb P(U_i>u)), \label{eq:vij12}\\
		&\mathbb P(U_i>u,U_j<-u^a) = o(\mathbb P(U_i>u)). \label{eq:vij13}
	\end{align} 
	
	Then we consider $A_{b,u}^{(1)}=\{U_b/B>(1+\delta_u)u, T_{\text{CRPT}}>u\}$ and $A_{b,u}^{(2)}=\{ U_b/B\le(1+\delta_u)u,T_{\text{CRPT}}>u\}$, where constant $\delta_u>0$ satisfies that $\delta_u\asymp u^{-1/2}$ as $t\to\infty$. Let $A_u^{(1)} = \bigcup_{b=1}^BA_{b,u}^{(1)}$ and $A_u^{(2)} = \bigcap_{b=1}^BA_{b,u}^{(2)}$. 
	Notice that since $T_{\text{CRPT}} = \sum_{b=1}^B U_b/B$, we have
	\begin{align*}
	\mathbb P(T_{\text{CRPT}}>u) 
	= \mathbb P(A_u^{(1)}) + \mathbb P(A_u^{(2)}).
	\end{align*}
	From Eq. (4) in the Supplement of \cite{cct}, to establish $\mathbb P(A_u^{(2)}) = o(1/u)$, it suffices to prove for any $1\le i\not=j\le B$,
	\begin{align*}
	I := \mathbb P(u<U_i\le B(1-\delta_u)u, U_j>B\delta_u u/(B-1)) = o(1/u).
	\end{align*}
	From \eqref{eq:vij12}, we have $I \le \mathbb P(U_i>u, U_j>\delta_u u) = o(\mathbb P(U_i>u)) = o(1/u)$, and then $\mathbb P(A_u^{(2)}) = o(1/u)$ holds.
	
	By Bonferroni inequality \citep[Lemma 1]{cct}, we have
	\begin{align*}
		\sum_{b=1}^B \mathbb P(A_{b,u}^{(1)}) -
		\sum_{1\le i<j\le B}\mathbb P(A_{i,u}^{(1)}\cap A_{j,u}^{(1)})
		\le
		 \mathbb P(A_u^{(1)})
		\le
		\sum_{b=1}^B \mathbb P(A_{b,u}^{(1)}).
	\end{align*}
	Notice that $\mathbb P(A_{i,u}^{(1)}\cap A_{j,u}^{(1)}) \le \mathbb P(U_i>B(1+\delta_u)u, U_j>B(1+\delta_u)u) = o(1/u)$.
	And
	\begin{align*}
		\mathbb P(A_{i,u}^{(1)}) 
		& = \mathbb P(U_b> B(1+\delta_u)u) - \mathbb P(U_b> B(1+\delta_u)u, T_{\text{CRPT}}\le u) \\
		& = \frac{1}{B\pi u} - \mathbb P(U_b> B(1+\delta_u)u, T_{\text{CRPT}}\le u) + o(1/u).
	\end{align*}
	The event $\{U_b> B(1+\delta_u)u, T_{\text{CRPT}}\le u\}$ imples that there exists at least one $j\not=b$ such that $U_j\le -B\delta_u u/(B-1)$. Therefore, with \eqref{eq:vij13}, we have 
	\begin{align*}
		\mathbb P(U_b>B(1+\delta_u)u,\ T_{\rm CRPT}\le u)
		\le \sum_{j\ne b}\mathbb P\left(
		U_b>B(1+\delta_u)u,\,
		U_j\le -\frac{B\delta_u u}{B-1}\right) = o(1/u).
	\end{align*}
	And then 
	\begin{align*}
	\mathbb P(T_{\text{CRPT}}>u) 
	= \mathbb P(A_u^{(1)}) + \mathbb P(A_u^{(2)})
	= 1/(\pi u) + o(1/u) \sim \mathbb P(W_C>u),
	\end{align*}
	which completes the proof.
\end{proof}

\subsubsection{Technical Lemmas for Theorem \ref{theo:null1}}

\begin{lemm}\label{lemm:t_quasi_hd}
	Suppose Conditions \ref{cond:normal}-\ref{cond:eigenvalues} hold and $H_0$ is true.
	Let $T_{i}^2$ and $T_{j}^2$ be the statistics as \eqref{eq:t_statistics} corresponding to the random projection matrices $\mathbf{P}_{i}$ and $\mathbf{P}_{j}$, $1\le i\not=j\le B$.
	
	(i) Let $t_2\asymp t_1^a$ for $a\in(0,1)$. Then, for fixed $n$, as $t_1\to\infty$ and $p/\log t_1\to\infty$,
	\begin{align}
	\mathbb P(T_{i}^2>t_1,T_{j}^2>t_2)
	=
	o\{\mathbb P(T_{i}^2>t_1)\}.
	\end{align}
	
	(ii) Let $t_3\asymp t_1^{-a(n-k+1)/k}$
	for $a\in(0,1)$. Then, for fixed $n$, as $t_1\to\infty$ and $p/\log t_1\to\infty$,
	\begin{align}
	\mathbb P(T_{i}^2>t_1,T_{j}^2<t_3)
	=
	o\{\mathbb P(T_{i}^2>t_1)\}.
	\end{align}
\end{lemm}

\begin{proof}
	Since the projection matrices are independent, $T_{i}^2$ and $T_{j}^2$ are independent conditional on $\mathcal D$. Therefore, we have
	\begin{align}\label{eq:lemmt_titj}
	\mathbb P(T_{i}^2\in I_1,T_{j}^2\in I_2)
	=
	\mathbb E\left[
	\mathbb P(T_{i}^2\in I_1\mid\mathcal D)
	\mathbb P(T_{j}^2\in I_2\mid\mathcal D)
	\right]
	\end{align}
	for any Borel sets $I_1,I_2$.

	{\bf Proof of Lemma \eqref{lemm:t_quasi_hd} (i):}

	Let 
	$0<\beta_1,\beta_2<1$ be chosen such that
	$(1-\beta_1)+a(1-\beta_2)>1$.
	By Lemma~\ref{lemm:cond_tail_hd} (i), for
	$\mathcal D\in\mathcal G_{p,n}$ and $p\ge 4n$, as $t_1\to\infty$ and $t_2\to\infty$, we have
	\begin{align*}
	&\mathbb P(T_{i}^2>t_1\mid\mathcal D)
	=O\big[
	\exp\{-t_1^{\beta_1}/8\}
	+t_1^{-(1-\beta_1)(n-k+1)/2}\big]
	=O\big[t_1^{-(1-\beta_1)(n-k+1)/2}\big],\\
	&\mathbb P(T_{j}^2>t_2\mid\mathcal D)
	=O\big[
	\exp\{-t_2^{\beta_2}/8\}
	+t_2^{-(1-\beta_2)(n-k+1)/2}\big]
	=O\big[t_2^{-(1-\beta_2)(n-k+1)/2}\big].
	\end{align*}
	Since $t_2\asymp t_1^a$, for
	$\mathcal D\in\mathcal G_{p,n}$, it follows that
	\begin{align}\label{eq:lemmt_td}
	\mathbb P(T_{i}^2>t_1\mid\mathcal D)\mathbb P(T_{j}^2>t_2\mid\mathcal D)
	= O\big(
	t_1^{-\{(1-\beta_1)+a(1-\beta_2)\}(n-k+1)/2}\big).
	\end{align}
	By Lemma~\ref{lemm:hd_good}, we have $\mathbb P(\mathcal G_{p,n}^c)\le 4e^{-p/32}$. 
	With \eqref{eq:lemmt_titj} and \eqref{eq:lemmt_td}, it follows that
	\begin{align}\label{eq:lemmt_titjo}
		\mathbb P(T_{i}^2>t_1,T_{j}^2>t_2)
		=&	\mathbb E\left[
		\mathbb P(T_{i}^2>t_1\mid\mathcal D)
		\mathbb P(T_{j}^2>t_2\mid\mathcal D)
		\right] \notag\\
		\le&
		\sup_{\mathcal D\in\mathcal G_{p,n}}
		\mathbb P(T_{i}^2>t_1\mid\mathcal D)\mathbb P(T_{j}^2>t_2\mid\mathcal D)
		+
		\mathbb P(\mathcal G_{p,n}^c) \notag\\
		=&O\big(
		t_1^{-\{(1-\beta_1)+a(1-\beta_2)\}(n-k+1)/2}+e^{-p/32}\big).
	\end{align}
	Meanwhile, by Lemma~\ref{lemm:t_tail}, we have
	$\mathbb P(T_{i}^2>t_1)
	\asymp t_1^{-(n-k+1)/2}$. 
	Recall that $(1-\beta_1)+a(1-\beta_2)>1$ and $p/\log t_1\to\infty$, with \eqref{eq:lemmt_titjo}, we have
	\begin{align*}
	\mathbb P(T_{i}^2>t_1,T_{j}^2>t_2)
	=o\big(t_1^{-(n-k+1)/2}\big) = o\big(\mathbb P(T_{i}^2>t_1)\big),
	\end{align*}
	which completes the proof of Lemma \eqref{lemm:t_quasi_hd} (i).

	{\bf Proof of Lemma \eqref{lemm:t_quasi_hd} (ii):}
	
	Let $0<\beta_3<1$, $0<\beta_4<1$ such that $a\beta_4>\beta_3$.
	Similarly, by Lemma~\ref{lemm:cond_tail_hd}, for
	$\mathcal D\in\mathcal G_{p,n}$ and $p\ge 4n$, as $t_1\to\infty$ and $t_3\to 0^+$, we have
	\begin{align*}
		\mathbb P(T_{i}^2>t_1\mid\mathcal D)
		=O\big[t_1^{-(1-\beta_3)(n-k+1)/2}\big]
		\quad {\rm and}\quad \mathbb P(T_{j}^2<t_3\mid\mathcal D)
		=O\big[t_3^{\beta_4k/2}\big].
	\end{align*}
	Since $t_3\asymp t_1^{-a(n-k+1)/k}$,
	it follows that $t_3^{\beta_4 k/2}
	\asymp t_1^{-a\beta_4 (n-k+1)/2}$.
	Thus, for $\mathcal D\in\mathcal G_{p,n}$, we have
	\begin{align}\label{eq:lemmt_titj2}
	\mathbb P(T_{i}^2>t_1\mid\mathcal D)\mathbb P(T_{j}^2<t_3\mid\mathcal D)
	=	O\big(
	t_1^{-(1-\beta_3+a\beta_4)(n-k+1)/2}	\big).
	\end{align}
	By Lemma~\ref{lemm:hd_good}, we have $\mathbb P(\mathcal G_{p,n}^c)\le 4e^{-p/32}$. 
	With \eqref{eq:lemmt_titj} and \eqref{eq:lemmt_titj2}, it follows that
	\begin{align}\label{eq:lemmt_titjo2}
		\mathbb P(T_{i}^2>t_1,T_{j}^2<t_3)
		=&	\mathbb E\left[
		\mathbb P(T_{i}^2>t_1\mid\mathcal D)
		\mathbb P(T_{j}^2<t_3\mid\mathcal D)
		\right] \notag\\
		\le&
		\sup_{\mathcal D\in\mathcal G_{p,n}}
		\mathbb P(T_{i}^2>t_1\mid\mathcal D)\mathbb P(T_{j}^2<t_3\mid\mathcal D)
		+
		\mathbb P(\mathcal G_{p,n}^c) \notag\\
		=&O\big(
		t_1^{-(1-\beta_3+a\beta_4)(n-k+1)/2}+e^{-p/32}\big).
	\end{align}
	Recall that $a\beta_4>\beta_2$ and then
	$1-\beta_2+a\beta_4>1$.
	Hence, as $t_1\to\infty$ and $p/\log t_1\to\infty$, we obtain
	\begin{align*}
	\mathbb P(T_{i}^2>t_1,T_{j}^2<t_3)
	=   o(t_1^{-(n-k+1)/2})
	=	o\{\mathbb P(T_{i}^2>t_1)\},
	\end{align*}
	which completes the proof.
\end{proof}

\begin{lemm}\label{lemm:u_quasi_hd}
	Suppose Conditions \ref{cond:normal}-\ref{cond:eigenvalues} hold and $H_0$ is true. 
	Let $T_{i}^2$ and $T_{j}^2$ be the Hotelling statistics corresponding to two independent random projection matrices $\mathbf P_{i}$ and $\mathbf P_{j}$, $1\le i \not= j\le B$. 
	Let $U_b=\phi(p_b)=\tan\{(0.5-p_b)\pi\}$ for $b=1,\ldots,B$, where $p_b$ denotes the $p$-value corresponding to $T_{b}^2$. 
	Then, for any $a\in(0,1)$ and fixed $n$, as $u\to\infty$ and $p/\log u\to\infty$, we have
	\begin{align}
	&\mathbb P(U_i>u,U_j>u^a)
	=
	o\{\mathbb P(U_i>u)\}, \\
	&\mathbb P(U_i>u,U_j<-u^a)
	=
	o\{\mathbb P(U_i>u)\}.
	\end{align}
\end{lemm}

\begin{proof}
	Under null hypothesis, we have $(n-k+1)T_{b}^2/(nk)$ follows $F_{k,n-k+1}$.
	Thus, $p_b$ is uniform on $(0,1)$, and $U_b=\phi(p_b)$ follows the standard Cauchy distribution. For $t>0$, define
	\[
	g(t):=\tan\left[\pi\left\{
	F_{k,n-k+1}\left(\frac{n-k+1}{nk}t\right)
	-\frac12	\right\}\right].
	\]
	Since $p_b=1-F_{k,n-k+1}((n-k+1)T_{b}^2/(nk))$,
	we have
	\begin{align}
	U_b	= \phi(p_b) = g(T_{b}^2).
	\end{align}
	Notice that the function $g:(0,\infty)\to\mathbb R$ is continuous and strictly increasing. Moreover,
	as $t\to\infty$, $g(t)\asymp t^{(n-k+1)/2}$,
	and as $t\to0^+$, $g(t)\asymp -t^{-k/2}$.
	Let $r:=g^{-1}$ as the inverse function of $g$. Then $r(u)\asymp u^{2/(n-k+1)}$ as $u\to\infty$
	and $r(-u)\asymp u^{-2/k}$ as $u\to\infty$.

	Taking $t_1=r(u)$ and $t_2=r(u^a)$, we have $t_2\asymp t_1^a$.
	Therefore, Lemma~\ref{lemm:t_quasi_hd} (i) gives
	\[
	\mathbb P(U_i>u,U_j>u^a)
	= \mathbb P(T_{i}^2>t_1,T_{j}^2>t_2)
	= o\{\mathbb P(T_{i}^2>t_1)\}
	= o\{\mathbb P(U_i>u)\}.
	\]
	Similarly, taking $t_3=r(-u^a)$, we have $t_3\asymp u^{-2a/k}$.
	Notice that since $t_1\asymp u^{2/(n-k+1)}$, it follows that $t_3\asymp t_1^{-a(n-k+1)/k}$.
	Therefore, as $u\to\infty$ and $p/\log u\to\infty$, we have $t_1\to\infty$ and $p/\log t_1\to\infty$.
	By Lemma~\ref{lemm:t_quasi_hd} (ii), it follows that
	\[
	\mathbb P(U_i>u,U_j<-u^a)
	=\mathbb P(T_{i}^2>t_1,T_{j}^2<t_3)
	=o\{\mathbb P(T_{i}^2>t_1)\}
	=o\{\mathbb P(U_i>u)\},
	\]
	which completes the proof.
\end{proof}

\subsection{Proof of {\color{black}Lemma}~\ref{prop:power1}}

\subsubsection{Proof of {\color{black}Lemma}~\ref{prop:power1}}

\begin{proof}

Under alternative hypothesis $H_1$,
the statistic $T_{b}^2$ admits an non-central $F$ distribution as
\begin{align*}
	\frac{n-k+1}{nk}T_{b}^2\mid \mathbf{P}_{b},\boldsymbol\delta
	\sim
	F\left(k,n-k+1,\gamma_{b}\right),
\end{align*}
where the non-centrality parameter is given by
\begin{align}\label{eq:gamma_k}
	\gamma_{b}
	:=
	\frac{n_1 n_2}{n_1+n_2}
	\boldsymbol\delta^\top
	\mathbf{P}_{b}^\top
	\big(\mathbf{P}_{b}\boldsymbol\Sigma\mathbf{P}_{b}^\top\big)^{-1}
	\mathbf{P}_{b}\boldsymbol\delta 
	= \frac{n_1 n_2}{n_1+n_2}\Delta_{b}^2.
\end{align}
Since  $ n_1n_2/\{(n_1+n_2)k\} \to {\color{black}\pi_s(1-\pi_s)/\pi_d}$ and ${\color{black}\pi_s(1-\pi_s)\Delta_{b}^2/\pi_d} \stackrel{p}{\to} \gamma$,
we have $\gamma_{b}/k \stackrel{p}{\to} \gamma$ as $n\to\infty$. From Lemma \ref{lemm:f_noncentral}, as $n\to\infty$, we have 
\begin{align}\label{eq:fgamma0}
	\sqrt{n}\bigg(\frac{n-k+1}{nk}T_{b}^2 - (1+\gamma_{b}/k)\bigg) \stackrel{d}{\to} N(0,\sigma^2),
\end{align} 
where $\sigma^2=2(1+2\gamma)/{\color{black}\pi_d}+2(1+\gamma)^2/(1-{\color{black}\pi_d})$.

Let $G(\cdot)$ denote the conditional distribution function of
$(n-k+1)T_{b}^2/(nk)$ given $\mathbf P_{b}$ and $\boldsymbol\delta$. 
For any $\alpha\in(0,1)$, we have
\begin{align*}
	\mathbb P(p_b\le\alpha\mid \mathbf P_{b},\boldsymbol\delta)
	&=\mathbb P_\theta\Big(\frac{n-k+1}{nk}T_{b}^2
	\ge F_{k,n-k+1}^{-1}(1-\alpha)
	\mid \mathbf P_{b},\boldsymbol\delta\Big)\\
	&=1-G\big(F_{k,n-k+1}^{-1}(1-\alpha)\big),
\end{align*}
From \eqref{eq:fgamma0}, as $n\to\infty$, we have 
\begin{align*}
G(u) - \Phi\bigg(\frac{\sqrt{n}(u-(1+\gamma_{b}/k))}{\sigma}\bigg)
\to 0.
\end{align*}
From Remark \ref{remm:Fnull}, 
as $n\to\infty$, we have
\begin{align*}
	F_{k, n-k+1}^{-1}(1-\alpha) 
	&= 1 - \sqrt{\frac{2}{{\color{black}\pi_d}(1-{\color{black}\pi_d})n}}z_\alpha + o(n^{-1/2}).
\end{align*}
Therefore, under the alternative satisfying Condition \ref{cond:alternative_relaxed}, as $n\to\infty$, 
\begin{align*}
	\mathbb P(p_b\le\alpha\mid \mathbf P_{b},\boldsymbol\delta) 
- \Phi\bigg(\frac{(\sqrt{2/({\color{black}\pi_d}(1-{\color{black}\pi_d}))}z_\alpha+\sqrt{n}\gamma_{b}/k))}{\sigma}\bigg) 
\stackrel{p}{\to} 0,
\end{align*}
which completes the proof.
\end{proof}

\subsubsection{Technical Lemma for {\color{black}Lemma}~\ref{prop:power1}}

\begin{lemm}\label{lemm:f_noncentral}
Let $F_\gamma$ denote a random variable such that, conditional on $\gamma_{b}$,
\[
F_\gamma \mid \gamma_{b} \sim F(k,n-k+1,\gamma_{b}).
\]
Assume $k/n\to {\color{black}\pi_d} \in (0,1)$ and $\gamma_{b}/k \stackrel{p}{\to} \gamma\in [0,\infty)$.
As $n\to\infty$, we have
\begin{align}\label{eq:fgamma}
	\sqrt{n}(F_\gamma - (1+\gamma_{b}/k)) \stackrel{d}{\to} N(0,\sigma^2),
\end{align} 
where $\sigma^2=2(1+2\gamma)/{\color{black}\pi_d}+2(1+\gamma)^2/(1-{\color{black}\pi_d})$.
\end{lemm}

\begin{proof}
We consider two independent random variables $Z_1$ and $Z_2$ such that $kZ_1\mid \gamma_{b} \sim \chi^2_k(\gamma_{b})$ and $(n-k+1)Z_2 \sim \chi^2_{n-k+1}$, where $\chi^2_k(\gamma_{b})$ denotes the non-central chi-square distribution with non-central parameter $\gamma_{b}$.
Define $g(z_1,z_2) := z_1/z_2$ and then $F_\gamma = g(Z_1,Z_2)$.

For the non-central chi-square variable $kZ_1$, we have
$(kZ_1-(k+\gamma_{b}))/\sqrt{2(k+2\gamma_{b})}\mid\gamma_{b} \xrightarrow{d} N(0,1)$ as $k \to \infty$. 
Then $\sqrt{n}\left( Z_1 - \theta_1 \right)\mid\gamma_{b} \xrightarrow{d} N\left(0,\sigma_{1,n}^2\right)$ as $n\to\infty$, 
where $\theta_1 = 1+\gamma_{b}/k$ and $\sigma_{1,n}^2 = 2(1+2\gamma_{b}/k)/{\color{black}\pi_d}$. Notice that since  $\sigma_{1,n}^2 \stackrel{p}{\to} \sigma_1^2$ with $\sigma_1^2 = 2(1+2\gamma)/{\color{black}\pi_d}$, invoking the Slutsky's theorem and we obtain
\begin{align}\label{eq:z1}
\sqrt{n}\left( Z_1 - \theta_1 \right) \xrightarrow{d} N\left(0,\sigma_1^2\right), \quad {\rm as \ }n\to\infty.
\end{align}
Similarly, for the chi-square variable $(n-k+1)Z_2$, we have
\begin{align}\label{eq:z2}
\sqrt{n}\left( Z_2 - \theta_2 \right) \xrightarrow{d} N\left(0,\sigma_2^2\right), \quad {\rm as \ }n\to\infty,
\end{align}
where $\theta_2 = 1$ and $\sigma_2^2 = 2/(1-{\color{black}\pi_d})$.

Notice that $Z_1 - \theta_1 = O_p(n^{-1/2})$ and $Z_2 - \theta_2 = O_p(n^{-1/2})$. 
As $n\to\infty$, the first-order Taylor expansion of $g$ at 
$(\theta_1,\theta_2)$ yields
\begin{align*}
	g(Z_1,Z_2)
	&= g(\theta_1,\theta_2)
	+ \nabla g(\theta_1,\theta_2)^\top
	\begin{pmatrix}
		Z_1-\theta_1\\
		Z_2-\theta_2
	\end{pmatrix}
	+ o_p(n^{-1/2}),
\end{align*}
where $\nabla g(z_1,z_2)^\top=(1/z_2, -z_1/z_2^2)$. 
Therefore, as $n\to\infty$, we obtain
\begin{align*}
	\sqrt{n}(F_\gamma - (1+\gamma_{b}/k))
	&=\sqrt{n}(Z_1-\theta_1)-\sqrt{n}(1+\gamma_{b}/k)(Z_2-\theta_2)
	+ o_p(1) \\
	&=\sqrt{n}(Z_1-\theta_1)-\sqrt{n}(1+\gamma)(Z_2-\theta_2)
	+ o_p(1).
\end{align*}
With \eqref{eq:z1}-\eqref{eq:z2}, it follows that
$\sqrt{n}\big(F_\gamma-(1+\gamma_{b}/k)\big)
\xrightarrow{d}
N\big(0,\sigma^2\big)$ with
$\sigma^2=\sigma_1^2+(1+\gamma)^2 \sigma_2^2
=2(1+2\gamma)/{\color{black}\pi_d}+2(1+\gamma)^2/(1-{\color{black}\pi_d})$, which completes the proof.

\end{proof}

\begin{remark}\label{remm:Fnull}
For Lemma \ref{lemm:f_noncentral}, if $\gamma_{b}=0$, we have
$\sigma^2 = 2/{\color{black}\pi_d} + 2/(1-{\color{black}\pi_d}) = 2/({\color{black}\pi_d}(1-{\color{black}\pi_d}))$,
and then
$\sqrt{{\color{black}\pi_d}(1-{\color{black}\pi_d})n/2} (F_0 - 1) \xrightarrow{d} N(0,1)$,
which coincides with the classical asymptotic normal approximation for the central $F$ distribution.
\end{remark}

\subsection{Proof of Lemma~\ref{lemm:delta_projection_bound}}

\begin{proof}
	Let
	\[
	\Pi_b
	=
	\mathbf P_{b}^\top
	(\mathbf P_{b}\mathbf P_{b}^\top)^{-1}
	\mathbf P_{b}.
	\]
	Then $\Pi_b$ is the orthogonal projection matrix onto the row space of $\mathbf P_{b}$.
	Notice that since $\lambda_{\min}(\boldsymbol\Sigma)\mathbf I_p
	\preceq
	\boldsymbol\Sigma
	\preceq
	\lambda_{\max}(\boldsymbol\Sigma)\mathbf I_p$, multiplying by $\mathbf P_{b}$ and $\mathbf P_{b}^\top$ gives
	\[
	\lambda_{\min}(\boldsymbol\Sigma)
	\mathbf P_{b}\mathbf P_{b}^\top
	\preceq
	\mathbf P_{b}\boldsymbol\Sigma\mathbf P_{b}^\top
	\preceq
	\lambda_{\max}(\boldsymbol\Sigma)
	\mathbf P_{b}\mathbf P_{b}^\top.
	\]
	Therefore,
	\begin{align}\label{eq:lemm_delta}
		\frac{1}{\lambda_{\max}(\boldsymbol\Sigma)}
		\frac{\boldsymbol\delta^\top\Pi_b\boldsymbol\delta}
		{\|\boldsymbol\delta\|_2^2}
		\le
		\frac{\Delta_{b}^2}{\|\boldsymbol\delta\|_2^2}
		\le
		\frac{1}{\lambda_{\min}(\boldsymbol\Sigma)}
		\frac{\boldsymbol\delta^\top\Pi_b\boldsymbol\delta}
		{\|\boldsymbol\delta\|_2^2}.
	\end{align}

	We next identify the distribution of $B_{k,p}:=\boldsymbol\delta^\top\Pi_b\boldsymbol\delta/\|\boldsymbol\delta\|_2^2$. 
	For any fixed orthogonal matrix $\mathbf O\in\mathbb R^{p\times p}$, each row of 
	$\mathbf P\mathbf O$ is still a $p$-dimensional standard normal vector. Hence
	$\mathbf P\mathbf O$ follows the same distribution as $\mathbf P$.
	Let $\mathbf u=\boldsymbol\delta/\|\boldsymbol\delta\|_2$.
	Choose an orthogonal matrix $\mathbf O$ such that $\mathbf O e_1=\mathbf u$, where 
	$e_1=(1,0,\ldots,0)^\top$. Then the distribution of
	$\mathbf u^\top \Pi_b \mathbf u$ is the same as that of
	$e_1^\top \Pi_b e_1$.
	Therefore, without loss of generality, we may take
	$\boldsymbol\delta=\|\boldsymbol\delta\|_2 e_1$.
	Now apply the Gram-Schmidt orthogonalization to the rows of $\mathbf P_{b}$. Let
	$q_1,\ldots,q_k\in\mathbb R^p$ be the orthonormal row vectors. Thus, we have $\Pi=\sum_{i=1}^k q_i q_i^\top$ and then
	\[
	B_{k,p}
	=
	e_1^\top\Pi e_1
	=
	\sum_{i=1}^k q_{i 1}^2,
	\]
	where $q_{i 1}$ denotes the first component of $q_i$.

	Let $Z_1,\ldots,Z_p$ be independent standard normal random variables, and then $(q_{11},q_{21},\ldots,q_{k1})$ follows the same distribution as 
	$(Z_1,\ldots,Z_k)^\top/\sqrt{Z_1^2+\cdots+Z_p^2}$.
	Therefore,
	\[
	B_{k,p}
	\stackrel d=
	\frac{Z_1^2+\cdots+Z_k^2}
	{Z_1^2+\cdots+Z_p^2}\sim {\rm Beta}\left(\frac{k}{2},\frac{p-k}{2}\right).
	\]

	Let $B_{k,p} = X/(X+Y)$, where $X$ and $Y$ are independent with $X\sim \chi^2(k)$ and $Y\sim \chi^2(p-k)$. For any $x<k/p$,
	$\{B_{k,p}\le x\}
	=
	\{(1-x)X-xY\le0\}$ holds.
	For any $\lambda>0$ with $\lambda<1/(2x)$, Chernoff's inequality gives
	\[
	\begin{aligned}
		\mathbb P(B_{k,p}\le x)
		&\le
		\mathbb E\exp\{-\lambda[(1-x)X-xY]\}=
		\{1+2\lambda(1-x)\}^{-k/2}
		\{1-2\lambda x\}^{-(p-k)/2}.
	\end{aligned}
	\]
	Let $r = k/p$. Optimizing over $\lambda$ gives
	$\mathbb P(B_{k,p}\le x)\le\exp\left\{-p/2D(r\|x)\right\}$,
	where
	\[
	D(r\|x)
	=
	r\log\frac{r}{x}
	+
	(1-r)\log\frac{1-r}{1-x}.
	\]
	Taking $x=b_1r$ gives
	\begin{align}\label{eq:lemm_b<}
		\mathbb P(B_{k,p}\le b_1r)
		\le
		\exp\left\{
		-\frac p2D(r\|b_1r)
		\right\}.
	\end{align}
	
	Similarly, for $x>r$ with $x<1$,
	$\{B_{k,p}\ge x\}=\{(1-x)X-xY\ge0\}$.
	For any $\lambda>0$ with $\lambda<1/\{2(1-x)\}$, we have
	\begin{align*}
		\mathbb P(B_{k,p}\ge x)\le
		\mathbb E\exp\{\lambda[(1-x)X-xY]\}=
		\{1-2\lambda(1-x)\}^{-k/2}
		\{1+2\lambda x\}^{-(p-k)/2}.
	\end{align*}
	Optimizing over $\lambda$ gives
	$\mathbb P(B_{k,p}\ge x)\le\exp\left\{-p/2D(r\|x)\right\}$.
	Taking $x=b_2r$, it follows that
	\begin{align}\label{eq:lemm_b>}
		\mathbb P(B_{k,p}\ge b_2r)
		\le
		\exp\left\{
		-\frac p2D(r\|b_2r)
		\right\}.
	\end{align}
	
	Since $0<b_1<1<b_2$ are fixed and $k\to\infty$, we have
	$pD(r\|b_1r)\to\infty$ and $pD(r\|b_2r)\to\infty$.
	With \eqref{eq:lemm_b<} and \eqref{eq:lemm_b>}, $\mathbb P(b_1r\le B_{k,p}\le b_2r)\to1$ holds.
	Consequently, combining this with \eqref{eq:lemm_delta}, we obtain
	\[
	\mathbb P\left(
	\frac{b_1 k}{p\lambda_{\max}(\boldsymbol\Sigma)}
	\|\boldsymbol\delta\|_2^2
	\le
	\Delta_{b}^2
	\le
	\frac{b_2 k}{p\lambda_{\min}(\boldsymbol\Sigma)}
	\|\boldsymbol\delta\|_2^2
	\right)
	\to1,
	\]
	which completes the proof.
\end{proof}

\subsection{Proof of Theorem~\ref{theo:power1}}

\begin{proof}
	By Lemma~\ref{lemm:delta_projection_bound}, for any fixed \(0<b_1<1\),
	\begin{align}\label{eq:theopower1}
		\mathbb{P} \bigg(\sqrt n\,\Delta_{b}^2
		\ge
		\frac{b_1}{\lambda_{\max}(\boldsymbol\Sigma)}
		\sqrt n\,\frac{k}{p}\|\boldsymbol\delta\|_2^2\bigg)\to 1.
	\end{align}
	Under Condition \ref{cond:c1} and Condition \ref{cond:alternative_relaxed}, we have \(k/n\to {\color{black}\pi_d}\in(0,1)\) and $p/n^{3/2}=o(\|\boldsymbol\delta\|_2^2)$. Hence, it follows that $\sqrt n\Delta_{b}^2\stackrel p\to\infty$.
	Let $\beta_\Delta(\Delta_{b}^2) = \mathbb{P}(p_b\le\alpha\mid \Delta_{b}^2)$ denote the power conditional on the projected signal $\Delta_{b}^2$, and then $\beta_\Delta(\Delta_{b}^2) = \beta(\mathbf{P}_{b},\boldsymbol{\delta})$.

	Firstly, we consider the weak projected signal $\underline{\Delta}^2 = \min \{b_1k\|\boldsymbol\delta\|_2^2/(p{\color{black}c_2}),n^{-1/4}\}$. 
	With \eqref{eq:theopower1}, $\mathbb{P}(\Delta_{b}^2\ge \underline{\Delta}^2)\to 1$ holds.
	Notice that since $\underline{\Delta}^2\to 0$, from {\color{black}Lemma} \ref{prop:power1}, we have
	\begin{align}\label{eq:theopower1_delta}
		\beta_\Delta(\underline{\Delta}^2)\to 1.
	\end{align}
	Since the conditional power function $\beta_\Delta(\underline{\Delta}^2)$ is monotone increasing in the parameter $\Delta_{b}^2$, for $\Delta_{b}^2 \ge \underline{\Delta}^2$, we have $\beta_\Delta(\Delta_{b}^2) \ge \beta_\Delta(\underline{\Delta}^2)$.
	Therefore, it follows that
	\begin{align*}
		\mathbb{P}(p_b\le\alpha )
		= \mathbb{E}_\Delta [\beta_{\Delta}(\Delta_{b}^2)] 
		\ge \mathbb{E}_\Delta [\beta_{\Delta}(\Delta_{b}^2)\mathbf{1}(\Delta_{b}^2\ge \underline{\Delta}^2)] 
		\ge  \beta_{\Delta}(\underline{\Delta}^2)\mathbb{P}(\Delta_{b}^2\ge \underline{\Delta}^2)).
	\end{align*}
	Notice that since
	$\mathbb{P}(\Delta_{b}^2\ge \underline{\Delta}^2) \to 1$, with \eqref{eq:theopower1_delta}, we obtain
	\begin{align*}
		\mathbb{P}(p_b\le\alpha )
		\ge  \beta_\Delta(\underline{\Delta}^2)\mathbb{P}(\Delta_{b}^2 \ge \underline{\Delta}^2) \ge \beta_\Delta(\underline{\Delta}^2)(1-o(1))
		\to 1.
	\end{align*}
	which completes the proof.
\end{proof}

{
\color{black}
\subsection{Proof of Proposition \ref{prop:power_comparison}}

\begin{proof}
By Condition~\ref{cond:conditional_power}, there exist constants
\(0<\eta<\alpha<1-\zeta<1\) such that
\[
\frac{\phi(\eta)+(B-1)\phi(1-\zeta)}{B}\ge t_\alpha
\quad {\rm and} \quad
q_\eta(\mathcal D)^B
+
Bq_{1-\zeta}(\mathcal D)
=
o_p\{q_\alpha(\mathcal D)\}.
\]
Since \(\phi(p)\) is strictly decreasing on \((0,1)\),
on the event
\[
\min_{1\le b\le B}p_b\le \eta
\quad\text{and}\quad
\max_{1\le b\le B}p_b\le 1-\zeta,
\]
we have
\[
T_{\rm CRPT}
=
\frac1B\sum_{b=1}^B\phi(p_b)
\ge
\frac{\phi(\eta)+(B-1)\phi(1-\zeta)}{B}
\ge
t_\alpha .
\]
Therefore,
\[
\{T_{\rm CRPT}< t_\alpha\}
\subset
\left\{\min_{1\le b\le B}p_b>\eta\right\}
\cup
\left\{\max_{1\le b\le B}p_b>1-\zeta\right\}.
\]
Taking conditional probabilities given \(\mathcal D\), we obtain
\begin{align*}
\mathbb P(T_{\rm CRPT}< t_\alpha\mid\mathcal D)
&\le
\mathbb P\left(\min_{1\le b\le B}p_b>\eta\mid\mathcal D\right)
+
\mathbb P\left(\max_{1\le b\le B}p_b>1-\zeta\mid\mathcal D\right).
\end{align*}
Conditional on \(\mathcal D\), the projected \(p\)-values are independent because the random projection matrices are independent. Hence,
\[
\mathbb P\left(\min_{1\le b\le B}p_b>\eta\mid\mathcal D\right)
=
q_\eta(\mathcal D)^B.
\]
Moreover, we notice that
\[
\mathbb P\left(\max_{1\le b\le B}p_b>1-\zeta\mid\mathcal D\right)
\le
Bq_{1-\zeta}(\mathcal D).
\]
Thus, we have
\[
\mathbb P(T_{\rm CRPT}< t_\alpha\mid\mathcal D)
\le
q_\eta(\mathcal D)^B+Bq_{1-\zeta}(\mathcal D).
\]
By Condition~\ref{cond:conditional_power},
\[
q_\eta(\mathcal D)^B+Bq_{1-\zeta}(\mathcal D)
=
o_p\{q_\alpha(\mathcal D)\}.
\]
Recall that $q_\alpha(\mathcal D)=\mathbb P(p_b>\alpha\mid\mathcal D)$,
we have
\[
\mathbb P(T_{\rm CRPT}< t_\alpha\mid\mathcal D)
=
o_p\left\{
\mathbb P(p_b>\alpha\mid\mathcal D)
\right\},
\]
which completes the proof.
\end{proof}
}

\subsection{Proof of Theorem~\ref{theo:power2}}

\begin{proof}
Recall that the Cauchy combination test statistic is $T_{\text{CRPT}} =\sum_{b=1}^B\phi(p_b)/B$. 
Note that $\phi(p)=\tan(\pi(0.5-p))$ is strictly decreasing on $(0,1)$, so that for  a fixed significance level $\alpha\in(0,1)$, we obtain $\{p_{b}\le\alpha\}=\{\phi(p_{b})\ge t_\alpha\}$.  
From Theorem \ref{theo:power1}, 
we have $\mathbb{P}(p_{b}\le\alpha)\to 1$ for each $b=1,\ldots,B$, which implies $\mathbb{P}(\phi(p_{b})\ge t_\alpha)\to 1$ as $n\to\infty$.
Since $B$ is fixed, it follows that
\begin{align*}
\mathbb{P}\bigg(\bigcap_{b=1}^B\{\phi(p_{b})\ge t_\alpha\}\bigg)\to 1.
\end{align*}
On the event $\mathcal{E} = \bigcap_{b=1}^B\{\phi(p_{b})\ge t_\alpha\}$, we have
\[
T_{\text{CRPT}} = \frac{1}{B}\sum_{b=1}^B\phi(p_{b}) \ge \frac{1}{B}\sum_{b=1}^B t_\alpha = t_\alpha.
\]
Hence
\[
\mathbb{P}(T_{\text{CRPT}}\ge t_\alpha) \ge \mathbb{P}\bigg(\bigcap_{b=1}^B\{\phi(p_{b})\ge t_\alpha\}\bigg) \to 1,
\]
which completes the proof.
\end{proof}

\subsection{Proof of Theorem~\ref{theo:pe}}

\subsubsection{Proof of Theorem~\ref{theo:pe}}

\begin{proof}
	Firstly, we define two events $\mathcal{E}_1$ and $\mathcal{E}_2$ as
	\begin{align*}
	&\mathcal{E}_1=\left\{\max _{1 \leq j \leq p}\left|\bar{x}_{ j}-\bar{y}_{j}-\delta_j\right| / s_{j j}^{1 / 2}<\delta_{n, p} / \sqrt{\left(n_1 n_2\right) / (n_1+n_2)}\right\},\\ &\mathcal{E}_2=\left\{\frac{4}{9} \leq s_{j j} / \sigma_{j j} \leq \frac{9}{4}, \forall j \in\{1, \ldots, p\}\right\} .
	\end{align*}
	For Lemma \ref{lemm:pe}, we have $\mathbb{P}(\mathcal{E}_1 \cap \mathcal{E}_2) \to 1$ 
	under Condition \ref{cond:normal}
	as $n,p\to\infty$ with $\log p=o(n)$.
	Under null hypothesis $H_0:\boldsymbol{\mu}_1= \boldsymbol{\mu}_2$, it holds that
	\begin{align*}
	\mathbb{P}\left(J_0=0 \mid H_0\right)=\mathbb{P}\left(\max _{1 \leq j \leq p}\left\{\left|\bar{x}_{j}-\bar{y}_{j}\right| / s_{j j}^{1 / 2}\right\}<\delta_{n, p} / \sqrt{\left(n_1 n_2 / (n_1+n_2)\right)} \mid H_0\right) \rightarrow 1.
	\end{align*}
	
	Since the alternative ${\color{black}H_1}$ satisfying $\max_{1\le j\le p}\{\left|\delta_j\right|/\sigma_{j j}^{1 / 2}\}>3  \delta_{n, p} / \sqrt{\left(n_1 n_2\right) / (n_1+n_2)}$ considered in Condition \ref{cond:alt_sparse}, there exists $1\le j^*\le p$ such that $\left|\delta_{j^*}\right|/\sigma_{{j^*} {j^*}}^{1 / 2}>3  \delta_{n, p} / \sqrt{\left(n_1 n_2\right) / (n_1+n_2)}$.
	Under event $\mathcal{E}_1 \cap \mathcal{E}_2$, we have 
	\begin{align*}
	\frac{\left|\bar{x}_{j^*}-\bar{y}_{j^*}\right|}{s_{{j^*} {j^*}}^{1 / 2}} &\geq \frac{\left|\delta_{j^*}\right|-\left|\bar{x}_{{j^*}}-\bar{y}_{{j^*}}-\delta_{j^*}\right|}{s_{{j^*} {j^*}}^{1 / 2}} \\
	&\geq \frac{2|\delta_{j^*}|}{3 \sigma_{{j^*} {j^*}}^{1 / 2}}-\delta_{n, p} / \sqrt{\left(n_1 n_2\right) / (n_1+n_2)}>\delta_{n, p} / \sqrt{\left(n_1 n_2\right) / (n_1+n_2)} .
	\end{align*}
	Therefore, $\mathbb{P}(J_0=\sqrt n\mid {\color{black}H_1})\to 1$ as $n,p\to\infty$ with $\log p=o(n)$,
	which completes the proof.
\end{proof}

\subsubsection{Technical Lemma for Theorem~\ref{theo:pe}}

\begin{lemm}\label{lemm:pe}
Let $\mathcal{V}$ be a parameter space for the mean vectors.
Under the Condition \ref{cond:normal}, as $n, p \to \infty$, for $\log p = o(n)$ and any constant $c_-<1<c_+$, we have
\begin{align}\label{eq:s/sigma}
\inf_{\boldsymbol{\mu}_1, \boldsymbol{\mu}_2 \in \mathcal{V}} \mathbb{P}\left( c_- < \frac{s_{jj}}{\sigma_{jj}} < c_+, \; \forall j \in \{1,\dots,p\} \;\Big|\; \boldsymbol{\mu}_1, \boldsymbol{\mu}_2 \right) \to 1,
\end{align}
where $s_{jj}$ and $\sigma_{jj}$ are the $j$-th diagonal element of $\mathbf{S}$ and $\boldsymbol{\Sigma}$, respectively.
Meanwhile, we have
\begin{align}\label{eq:x-y}
	\inf_{\boldsymbol{\mu}_1, \boldsymbol{\mu}_2 \in \mathcal{V}} \mathbb{P}\left( \max_{1 \le j \le p} \frac{|\bar{x}_j - \bar{y}_j - \delta_j|}{\sqrt{s_{jj}}} < \frac{\delta_{n,p}}{\sqrt{n_1 n_2 / (n_1+n_2)}} \;\Big|\; \boldsymbol{\mu}_1, \boldsymbol{\mu}_2 \right) \to 1,
\end{align}
where $\bar{x}_j$ and $\bar{y}_j$ are the $j$-th component of $\bar{\mathbf{x}}$ and $\bar{\mathbf{y}}$, respectively.
\end{lemm}

\begin{proof}
	We prove \eqref{eq:s/sigma} and \eqref{eq:x-y}  in order.
	
	\noindent
	\textbf{Proof of \eqref{eq:s/sigma}.}
	Under Condition \ref{cond:normal}, for each $j$,
	$ns_{jj}/\sigma_{jj}
	\sim \chi^2_{n}$.
	From concentration bounds for Gaussian quadratic forms \citep[Lemma 1]{2000Laurent}, for any $u_1,u_2>0$, we have 
	\begin{align*}
		&\mathbb{P}\big(n(s_{jj}/\sigma_{jj} - 1) \ge   2\sqrt{nu_1} + 2u_1 \big) \le \exp\{-u_1\} ,  \\
		&\mathbb{P}\big(n(1 - s_{jj}/\sigma_{jj}) \ge   2\sqrt{nu_2}\big) \le \exp\{-u_2\}.
	\end{align*}
	For any
	$\epsilon\in(0,1)$, taking $u_1 = n\epsilon^2/4$ and $u_2 = n\epsilon^2/8$, we have
	\[
	\mathbb P\!\left(
	\left|\frac{s_{jj}}{\sigma_{jj}}-1\right|>\epsilon
	\right)
	\le
	2\exp(-n\epsilon^2/8).
	\]
	It follows that 
	\[
	\mathbb P\left(
	\max_{1\le j\le p}
	\left|\frac{s_{jj}}{\sigma_{jj}}-1\right|>\epsilon
	\right)
	= 	\mathbb P\bigg(
	\bigcup_{j=1}^p
	\Big\{\Big|\frac{s_{jj}}{\sigma_{jj}}-1\Big|>\epsilon\Big\}
	\bigg)
	\le
	2p\exp(-n\epsilon^2/8).
	\]
	where the right-hand side converges to zero for $\log p=o(n)$. Hence, 
	for any constants $0<c_-<1<c_+<\infty$ and $\boldsymbol{\mu}_1, \boldsymbol{\mu}_2 \in \mathcal{V}$, we have
	\[
	\mathbb P\left(
	c_-<\frac{s_{jj}}{\sigma_{jj}}<c_+,\ \forall j
	\right)\to1,
	\]
	which completes the proof.

	\noindent
	\textbf{Proof of \eqref{eq:x-y}.}
	Under Condition \ref{cond:normal}, for each $j=1,\ldots,p$, let
	\[
	Z_j=
	\frac{\bar x_j-\bar y_j-\delta_j}
	{\sqrt{\sigma_{jj}(n_1^{-1}+n_2^{-1})}}
	\sim N(0,1).
	\]
	Then for any $t>0$, by Gaussian tail inequality,
	$\mathbb P(|Z_j|>t)\le 2e^{-t^2/2}$ holds.
	It follows that
	\[
	\mathbb P\Big(\max_{1\le j\le p}|Z_j|>t\Big)
	= \mathbb P\bigg(\bigcup_{j=1}^p\{|Z_j|>t\}\bigg)
	\le 2p e^{-t^2/2}.
	\]
	Take $t=c_0\log(\log n)\sqrt{\log p}$.
	Then
	\[
	\mathbb P\Big(\max_{1\le j\le p}|Z_j|>t\Big)
	\le 2\exp\!\left(\log p-\frac{c_0^2}{2}(\log\log n)^2\log p\right)
	\to0 \quad {\rm as\ } n,p\to\infty. 
	\]
	Since
	$\delta_{n,p}=c_0\log(\log n)\sqrt{\log p}$, for any $\boldsymbol{\mu}_1, \boldsymbol{\mu}_2 \in \mathcal{V}$, we obtain
	\[
	\mathbb{P}\bigg(\max_{1\le j\le p}
	\frac{|\bar x_j-\bar y_j-\delta_j|}
	{\sqrt{\sigma_{jj}}}
	<
	\frac{\delta_{n,p}}{\sqrt{n_1n_2/(n_1+n_2)}}\bigg)\to 1 \quad {\rm as\ } n,p\to\infty.
	\]
	
	It remains to replace $\sigma_{jj}$ by $s_{jj}$. Under
	\eqref{eq:s/sigma}, for $\log (p) = o(n)$ we have
	$\sqrt{s_{jj}}\asymp \sqrt{\sigma_{jj}}$
	uniformly in $j$, which yields \eqref{eq:x-y}.
	
\end{proof}

\subsection{Proof of Theorem~\ref{theo:tpe}}

\begin{proof}
	By Theorem \ref{theo:power1}, $\mathbb{P}({\color{black}T_{\mathrm{CRPT\text{-}PE}}} \ge t_{\alpha}) \ge \mathbb{P}(T_{\rm CRPT} \ge t_{\alpha}) \to 1$ holds under the alternative satisfying Condition~\ref{cond:alternative_relaxed}. Thus, it suffices to prove  $\mathbb{P}({\color{black}T_{\mathrm{CRPT\text{-}PE}}} \ge t_{\alpha}) \to 1$ under the alternative ${\color{black}H_1'}$ satisfying Condition~\ref{cond:alt_sparse}.
		
	Notice that
	\begin{align}
		\mathbb P({\color{black}T_{\mathrm{CRPT\text{-}PE}}}\ge t_\alpha\mid {\color{black}H_1'})
		&=
		\mathbb P(T_{\rm CRPT}+J_0\ge t_\alpha\mid {\color{black}H_1'})\notag\\
		&\ge
		\mathbb P(T_{\rm CRPT}+\sqrt n\ge t_\alpha,\ J_0=\sqrt n\mid {\color{black}H_1'})\notag\\
		&\ge
		1-
		\mathbb P(J_0\ne\sqrt n\mid {\color{black}H_1'})
		-
		\mathbb P(T_{\rm CRPT}< t_\alpha-\sqrt n\mid {\color{black}H_1'}).
	\end{align}
	By Theorem~\ref{theo:pe}, $\mathbb P(J_0\ne\sqrt n\mid {\color{black}H_1'}) = 1-\mathbb P(J_0=\sqrt n\mid {\color{black}H_1'})\to0$ holds.
	Thus, it remains to show that $\mathbb P(T_{\rm CRPT}< t_\alpha-\sqrt n\mid {\color{black}H_1'})\to0$.
	
	Recall that $U_b=\phi(p_b)=\tan\{\pi(0.5-p_b)\}$ and	$T_{\rm CRPT}=\sum_{b=1}^B U_b/B$.
	Conditional on the projection matrix $\mathbf{P}_{b}$, the projected statistic $T^2_{b}$ has a
	noncentral \(F\) distribution under the alternative. Since the noncentrality
	parameter $\gamma_k$ in \eqref{eq:gamma_k} is nonnegative, the corresponding \(p\)-value is stochastically no
	larger than a uniform random variable on \((0,1)\). Equivalently, \(U_b\) is
	stochastically no smaller than a standard Cauchy random variable \(J\). Therefore,
	for any \(x>0\),
	\begin{align}\label{eq:theopower_ub}
	\mathbb P(U_b< -x\mid {\color{black}H_1'})
	\le
	\mathbb P(J< -x).
	\end{align}
	For fixed \(B\), we have
	\begin{align}\label{eq:theopower_tub}
		\mathbb P(T_{\rm CRPT}< t_\alpha-\sqrt n\mid {\color{black}H_1'})
		&\le 
		\mathbb P\bigg(\bigcup_{b=1}^B \{U_b< t_\alpha-\sqrt n \} \mid {\color{black}H_1'}\bigg) \notag \\
		&\le
		\sum_{b=1}^B
		\mathbb P(U_b< t_\alpha-\sqrt n\mid {\color{black}H_1'}).
	\end{align}
	Since \(t_\alpha\) is fixed and \(t_\alpha-\sqrt n\to-\infty\), it follows that
	$\mathbb P(J< t_\alpha-\sqrt n)\to0$.
	Thus, with \eqref{eq:theopower_ub}-\eqref{eq:theopower_tub}, we have
	$\mathbb P(T_{\rm CRPT}< t_\alpha-\sqrt n\mid {\color{black}H_1'})\to0$, which completes the proof.
\end{proof}

\newpage

\clearpage

\end{document}